\documentclass[11pt]{article}
\usepackage{amsmath,amsfonts,amsthm,amssymb}
\usepackage{fullpage}

\usepackage{longtable}
\usepackage{hyperref}
\usepackage[capitalise]{cleveref}
\usepackage{xspace}
\usepackage{ifthen}
\usepackage{float}
\usepackage{graphicx}
\usepackage{enumitem}
\usepackage{comment}
\usepackage{stmaryrd}
\usepackage{xcolor}
\usepackage{booktabs}
\usepackage{array}
\usepackage{tikz}
\usetikzlibrary{positioning, shapes, arrows.meta}
\usepackage{thm-restate}

\usetikzlibrary{calc} 

\newtheorem{theorem}{Theorem}[section]
\newtheorem{corollary}[theorem]{Corollary}
\newtheorem{lemma}[theorem]{Lemma}

\newtheorem{definition}[theorem]{Definition}
\newtheorem{fact}[theorem]{Fact}

\usepackage{bbm}

\usepackage[ruled,vlined,linesnumbered]{algorithm2e}
\crefname{algocf}{Algorithm}{Algorithms}

\newenvironment{lp*}{\begin{equation*}  \begin{array}{lll}}{\end{array}\end{equation*}}

\renewcommand{\emptyset}{\varnothing}
\newcommand{\vol}{\operatorname{vol}}

\newcommand{\protocol}[1]{\mathcal{P}^{\textnormal{#1}}}
\newcommand{\cycle}[0]{\mathsf{Cycle}}
\newcommand{\negcycle}[0]{\mathsf{NegativeCycle}}
\newcommand{\SSSP}[0]{\mathsf{SSSP}}
\newcommand{\negSSSP}[0]{\mathsf{NegativeSSSP}}

\usepackage{tikz}
\tikzset{
    node/.style={circle,draw=black},
    edge/.style={draw=black}
}

\usepackage[backend=biber, style=alphabetic, doi=false, url=false, isbn=false, minalphanames = 3, maxalphanames = 3, maxbibnames=10]{biblatex}
\usepackage{caption}

\title{
On the Communication Complexity of \\ Maximum Matching and Negative-Weight Shortest Paths}
\author{Yu Cheng\footnote{Brown University. \href{mailto:yu_cheng@brown.edu}{\url{yu_cheng@brown.edu}}. Supported in part by NSF Award CCF-2307106.} 
\and 
Tianle Jiang\footnote{Duke University. \href{mailto:tianle.jiang@duke.edu}{\url{tianle.jiang@duke.edu}}. Supported in part by NSF Award IIS-2402823. Part of the work was done while visiting Brown University. }
\and
Pachara Sawettamalya\footnote{Princeton University. \href{mailto:ps3122@princeton.edu}{\url{ps3122@princeton.edu}}. Supported in part by NSF CAREER Award CCF-2339942.}
\and 
Huacheng Yu\footnote{Princeton University. \href{mailto:yuhch123@gmail.com}{\url{yuhch123@gmail.com}}. Supported in part by NSF CAREER Award CCF-2339942.}
}
\date{}

\begin{document}

\maketitle

\begin{abstract}
We revisit several fundamental graph problems in the deterministic two-party communication model. Our main contributions include:
\begin{itemize}
    \item We give a new $\widetilde{O}(n^{3/2})$-bit protocol for computing a maximum matching in general graphs. While the same upper bound can be obtained by simulating the classic algorithms of Micali-Vazirani~\cite{MicaliV80} and Gabow~\cite{Gabow17}, our protocol is conceptually simple and avoids the intricacies of finding a maximal set of shortest augmenting paths.
    \item We give a new $\widetilde{O}(n)$-bit protocol for negative-cycle detection and negative-weight single-source shortest paths. Our protocol simplifies that of Blikstad et al.~\cite{BlikstadBEMN22} by replacing a long chain of reductions with a more direct approach based on vertex potentials.
    \item We give a combinatorial $\widetilde{O}(n)$-bit protocol for computing a maximum matching in bipartite graphs, obtained by reinterpreting the near-linear communication protocol of Blikstad et al.~\cite{BlikstadBEMN22} through a discretized analysis.
\end{itemize}
Together, these results provide simpler protocols for several basic graph problems.
We hope they will inspire further advances on the communication complexity of a wide range of graph problems.
\end{abstract}

\newpage

\section{Introduction}

Graph theory is central to computer science, both in theory and in practice. The field studies fundamental problems such as connectivity, shortest paths, matching, spanning forests, and cycle detection. Efficient solutions for these problems form the algorithmic backbone of many computational systems. While the complexity of these problems is well understood in centralized models, much less is known when the input is distributed across players with only partial views of the graph.

In this work, we study the \emph{communication complexity of graph problems} in the standard two-party model \cite{Yao79}. The edge set of a graph is partitioned between two players, Alice and Bob, who must collaboratively solve a global graph problem while minimizing the amount of communication between them. Although this model is closely related to query, streaming, and distributed models of computation, many fundamental graph problems remain poorly understood
from the perspective of communication complexity.

A recent breakthrough of Blikstad et al.~\cite{BlikstadBEMN22} showed that \emph{maximum bipartite matching} and several problems reducible to it admit deterministic protocols using only $\widetilde{O}(n)$ bits of communication.\footnote{Throughout this paper, we use $\widetilde{O}(\cdot)$ to hide $\mathrm{polylog}(n)$ factors.}
At the core of their approach is a cutting-plane framework testing whether a polytope (or feasible region of a linear program (LP)) is empty or has non-trivial volume. This framework is particularly well-suited for the two-party edge-partition model:
Given a graph problem $\mathcal{T}$, one formulates an LP whose feasibility gives the answer to $\mathcal{T}$, with constraints that are locally checkable and succinctly encodable for communication.
The cutting-plane framework is then used to determine the feasibility of this LP, and the result is translated back into an answer for $\mathcal{T}$.

As a concrete example, Blikstad et al.~\cite{BlikstadBEMN22} obtained a near-linear communication protocol for bipartite matching via the dual LP. Specifically, they reduce maximum bipartite matching to testing the feasibility of a vertex cover LP, which can be done using $O(n\log^2 n)$ bits of communication via the cutting-plane framework.

\subsection{Our Contributions}
\subsubsection{Maximum Matching in General Graphs}
Given the success of LP-based methods for bipartite graphs, a natural question is whether the same framework can be extended to \emph{general (non-bipartite) matching}. However, general graphs do not admit a polynomial-size LP formulation with the same clean matching/vertex-cover duality as in bipartite graphs. The standard (primal) LP for non-bipartite matching has exponentially many constraints, so its dual LP has exponentially many variables. This makes it difficult to generalize the volume-based analysis in \cite{BlikstadBEMN22}, which is highly sensitive to the dimension of the underlying polytope.

While achieving near-linear communication for general matching appears to require new ideas, we make partial progress by designing a simple protocol for general matching using $\widetilde{O}(n^{3/2})$ bits of communication.

\begin{theorem}[Informal]
There is a deterministic two-party communication protocol for maximum matching in general graphs using $\widetilde{O}(n^{3/2})$ bits.
\end{theorem}

This $\widetilde{O}(n^{3/2})$-bit upper bound is not entirely new: it can also be obtained by simulating the classic $O(m\sqrt n)$-time algorithms such as Micali-Vazirani~\cite{MicaliV80} and Gabow~\cite{Gabow17}. However, these algorithms are quite intricate even in the sequential setting. Although local computation is free in the communication model (which makes sophisticated data structures less of a bottleneck), turning them into explicit communication-efficient protocols still requires a careful implementation of the many steps in these algorithms.

Our main contribution lies in the \emph{simplicity} of both the protocol and its correctness proof. At a high level, our protocol combines a greedy approach based on the Tutte-Berge formula with a careful simulation of Edmonds' Blossom algorithm~\cite{Edmonds65}. For the protocol description, it suffices to focus on the problem of deciding whether the input graph has a perfect matching. In the first phase, we maintain the set of low-value Tutte-Berge covers (\Cref{def:tutte-berge-cover}) that cover the edges revealed so far, and greedily reveal edges that invalidate as many such covers as possible. The first phase uses $\widetilde{O}(n^{3/2})$ bits of communication, and either correctly decides that no perfect matching exists, or returns a matching of size $\frac{n}{2}-\widetilde{O}(\sqrt n)$. In the second phase, we simulate $\widetilde{O}(\sqrt n)$ rounds of the Blossom algorithm to augment this partial matching to a perfect matching, using another $\widetilde{O}(n^{3/2})$ bits of communication. It is worth noting that our greedy first phase is conceptually simple and more communication-efficient than directly simulating the Blossom algorithm from scratch (which would require $\widetilde{O}(n^2)$ bits).

\subsubsection{Beyond Bipartite Matching}
We revisit the cutting-plane approach of \cite{BlikstadBEMN22}.
By abstracting away the problem-specific aspects of its application to bipartite matching, we show that it provides a general framework for designing communication-efficient protocols for graph problems.

We consider the following four tasks:
\begin{itemize}
\item $\cycle$: Given an unweighted directed graph, decide whether it contains a directed cycle.
\item $\negcycle$: Given a directed graph with integer weights in $[-W, W]$, decide whether it contains a negative-weight cycle.
\item $\SSSP$: Given a directed graph with integer weights in $[0, W]$ and a source vertex $s$, compute the shortest-path distances from $s$ to every vertex $v$.
\item $\negSSSP$: Given a directed graph with integer weights in $[-W, W]$ and a source vertex $s$, either detect a negative-weight cycle or, if none exists, compute the shortest-path distances from $s$ to every vertex $v$.
\end{itemize}

\begin{theorem}[Informal]
There are deterministic two-party communication protocols using $O(n \cdot \mathrm{polylog}(n, W))$ bits of communication for $\cycle$, $\negcycle$, $\SSSP$, and $\negSSSP$.
\end{theorem}

Two of these problems admit straightforward near-linear communication protocols via simulations of classical algorithms: $\cycle$ can be solved using Kosaraju's algorithm, and $\SSSP$ can be solved using Dijkstra's algorithm. We discuss these simulations in \Cref{subsec:kosaraju} and \Cref{subsec:dijkstra}.

The remaining two problems, $\negcycle$ and $\negSSSP$, are more challenging due to the presence of negative edge weights. In the sequential setting, both can be solved using the Bellman-Ford algorithm. However, a straightforward simulation of Bellman-Ford requires $\widetilde{\Omega}(n^2)$ bits in the communication model. To avoid this cost, we apply the cutting-plane framework of \cite{BlikstadBEMN22} to a vertex-potential polytope (which encodes the triangle inequalities for shortest-path distances). This gives a near-linear communication protocol for $\negcycle$, as we discuss in \Cref{subsec:negcycle}. Moreover, a feasible point of this polytope is a key ingredient in our protocol for $\negSSSP$.

The near-linear communication upper bound for $\negSSSP$ is not new: it was established in \cite{BlikstadBEMN22}. Their approach, however, proceeds through a chain of reductions, ultimately reducing $\negSSSP$ to bipartite min-cost $b$-matching, which they solve using $\widetilde{O}(n)$ bits of communication (see \Cref{fig:reduction}). These layers of reductions make it difficult to extract a self-contained protocol for $\negSSSP$ from their work.

\usetikzlibrary{arrows.meta, positioning}

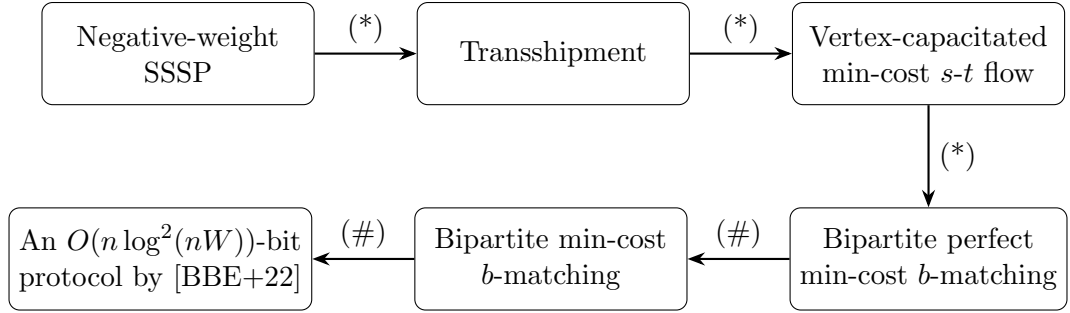
\begin{figure}[ht]
\begin{center}
\begin{tikzpicture}[
    node distance=1.35cm,
    problem/.style={draw, rounded corners, align=center, minimum width=3.6cm, minimum height=1.35cm},
    arrow/.style={-{Stealth}, thick}
]

\node[problem] (negSSSP)    {Negative-weight \\ SSSP};
\node[problem] (tranship)   [right=of negSSSP] {Transshipment};
\node[problem] (vcflow)     [right=of tranship] {Vertex-capacitated \\ min-cost $s$-$t$ flow};
\node[problem] (bpmcb)      [below=of vcflow] {Bipartite perfect \\ min-cost $b$-matching};
\node[problem] (bmcb)       [left=of bpmcb] {Bipartite min-cost \\ $b$-matching};
\node[problem] (prot)       [left=of bmcb] {An $O(n \log^2(nW))$-bit \\ protocol by \cite{BlikstadBEMN22}};


\draw[arrow] (negSSSP) -- node[midway, above] {(*)} (tranship);
\draw[arrow] (tranship) -- node[midway, above] {(*)} (vcflow);
\draw[arrow] (vcflow) -- node[midway, right] {(*)} (bpmcb);
\draw[arrow] (bpmcb) -- node[midway, above] {(\#)} (bmcb);
\draw[arrow] (bmcb) -- node[midway, above] {(\#)} (prot);


\end{tikzpicture}
\caption{A sequence of reductions that computes negative-weight SSSP in $O(n \log^2(nW))$ bits of communication. The reductions (*) are given in \cite{BrandLNPSS0W20} and (\#) are given in \cite{BlikstadBEMN22}.}
\label{fig:reduction}
\end{center}
\end{figure}

Our contribution is a conceptually simple and direct protocol for $\negSSSP$ that avoids these reductions. At a high level, the protocol computes a feasible point of the vertex-potential polytope used for $\negcycle$. We then use this point as a proxy for the vertex potentials that would arise in Johnson's algorithm for all-pairs shortest paths. Reweighting the edges using these potentials essentially yields a communication-efficient reduction from $\negSSSP$ to $\SSSP$. Combined with an efficient simulation of Dijkstra's algorithm, this allows us to recover the exact shortest-path distances using near-linear communication. We describe our approach in detail in \Cref{subsec:negSSSP}.

\subsubsection{Combinatorial Protocol for Bipartite Matching} 

We present a \emph{combinatorial} variant of the near-linear communication protocol for bipartite matching in \cite{BlikstadBEMN22}.
Instead of tracking the continuous volume of the vertex-cover polytope, our protocol tracks a discrete analogue: the number of feasible grid points, corresponding to $\frac{1}{\mathrm{poly}(n)}$-integral vertex covers.
We show that this discrete ``volume'' is a sufficiently accurate proxy for continuous volume in the analysis.
This gives a near-linear communication protocol for bipartite matching whose implementation is purely combinatorial.
We discuss this in \Cref{sec:bipartite_matching}.

\subsection{Two-Party Communication Complexity of Graph Problems}

\begin{longtable}{|p{5.26cm}|p{3cm}|p{6.93cm}|}
\hline
\textbf{Problem} & \textbf{Communication Complexity} & \textbf{Remarks} \\
\hline
\endfirsthead
\multicolumn{3}{l}{\textbf{Table \thetable \ Continued}} \\
\hline
\textbf{Problem} & \textbf{Communication Complexity} & \textbf{Remarks} \\
\hline
\endhead
\noalign{\vskip1em} 
\caption{Two-party communication complexities of various graph problems.}
\label{tab:communication_complexities}
\endlastfoot
BFS, DFS, spanning forest, connectivity, reachability & $\Theta(n \log n)$ & Direct simulation of classic algorithms; lower bound by \cite{HajnalMT88}.\\
\hline
$k$-edge connectivity & $O(kn \log n)$ & Nagamochi-Ibaraki algorithm \cite{NagamochiI92} \newline using $k$ spanning forests. \\
\hline
$s$-$t$ vertex connectivity & $\widetilde{O}(n)$ & A reduction to vertex-capacitated max flow \cite{BlikstadBEMN22}. \\
\hline
Global vertex connectivity & $\Theta(n^2)$ (det.) \newline $\widetilde{\Theta}(n^{3/2})$ (rand.) & {Combining local connectivity (det.); \newline \cite{BlikstadJMY25} (rand.). } \\
\hline
Bipartite maximum matching & $O(n \log^2 n)$ & Continuous LP-based protocol \cite{BlikstadBEMN22}. {\bf We give a combinatorial protocol in \Cref{subsec:comb_BMM}}. \\
\hline
General maximum matching & $\widetilde{O}(n^{3/2})$ & Simulation of \cite{MicaliV80, Gabow17}. {\bf We give a simple protocol in \Cref{sec:general_matching}}. \\
\hline
Strongly connected components & $O(n \log n)$ & Kosaraju's algorithm (see \Cref{subsec:kosaraju}). \\
\hline
Directed cycle detection & $O(n \log n)$ & Kosaraju's algorithm and DFS \newline (see \Cref{subsec:kosaraju}). \\
\hline
Negative cycle detection & $O(n \log^2(nW))$ & A chain of reductions \cite{BlikstadBEMN22}. {\bf We give a simple protocol in \Cref{subsec:negcycle}}. \\
\hline
Single-source shortest paths \newline (SSSP) & $O(n \log(nW))$ & Dijkstra's algorithm (see \Cref{subsec:dijkstra}). \\
\hline
Negative-weight SSSP & $O(n \log^2(nW))$ & Reducing to bipartite \newline min-cost $b$-matching \cite{BlikstadBEMN22}. {\bf We give a simple protocol in \Cref{subsec:negSSSP}}. \\
\hline
Global min-cut & $\widetilde{O}(n)$ & Derandomized reduction to cut-query \newline algorithms of \cite{RubinsteinSW18, ApersEGLMN22, KennethMordochK25}. \\
\hline
Exact min $s$-$t$ cut and undirected uncapacitated max flow & $\widetilde{O}(n^{11/7})$ (det.) \newline $\widetilde{O}(n^{3/2})$ (rand.) & \cite{JiangNS26} (det.); \newline \cite{GholizadehJ25, Kenneth-MordochK26, vdBrandSW26} (rand.). \\
\hline
Cut sparsification & $\widetilde{O}(n/\varepsilon^2)$ & Combining local cut sparsifiers of size $\widetilde{O}(n/\varepsilon^2)$ from both players. \\
\hline
$\varepsilon$-approximate min $s$-$t$ cut & $\widetilde{O}(n/\varepsilon^2)$ & Taking a min $s$-$t$ cut of a cut sparsifier. \\
\hline
Minimum spanning tree & $\widetilde{O}(n)$ & Simulating Kruskal's algorithm. \\
\hline
$(\Delta+1)$-coloring & $\widetilde{O}(n)$ (det.) \newline $\Theta(n)$ (rand.) & \cite{AssadiCGS23} (det.); \newline \cite{FlinM25} (rand.). \\
\hline
\end{longtable}

\section{Maximum Matching in General Graphs}
\label{sec:general_matching}

The problem of \emph{perfect matching in general graphs} ($\mathsf{PM}$) is defined in the communication model as follows: Alice and Bob are given a disjoint edge partition of a graph $G=(V,E)$ where $n = |V|$ is even.
Their goal is to determine whether $G$ has a perfect matching and, if so, to output one.

In this section, we present a deterministic $\widetilde O(n^{3/2})$-bit communication protocol for $\mathsf{PM}$.

\begin{theorem}
\label{thm:matching}
There is a deterministic protocol $\protocol{PM}$ that solves $\mathsf{PM}$ using $O(n^{3/2}\log^{3/2} n)$ bits of communication.
\end{theorem}

As a corollary, the protocol can be extended to compute a \emph{maximum matching} with an additional $O(\log n)$ factor in communication.

\begin{corollary}
There is a deterministic protocol $\protocol{MM}$ that computes a maximum matching
in $G$ using $O(n^{3/2}\log^{5/2} n)$ bits of communication.
\end{corollary}
\begin{proof}
For $0 \le r \le \frac{n}{2}$, let $G_r$ be the graph obtained by adding $n-2r$ new auxiliary vertices to $G$, connecting each auxiliary vertex to every original vertex (and no edges between auxiliary vertices).
Observe that $G_r$ has a perfect matching if and only if $G$ has a matching of size at least $r$.\footnote{Given a size-$r$ matching in $G$, the remaining $n-2r$ original vertices can be matched to the auxiliary vertices. Conversely, a perfect matching of $G_r$ matches each auxiliary vertex to a distinct original vertex, leaving exactly $2r$ original vertices to be matched among themselves.}
Thus, Alice and Bob can binary search over $r$ to find the size of a maximum matching in $G$ using $O(\log n)$ calls to $\protocol{PM}$ on graphs $G_r$ with at most $2n$ vertices.
For the final value of $r$, the original-original edges in a perfect matching of $G_r$ form a maximum matching of $G$.
By \Cref{thm:matching}, the total communication is $O(\log n) \cdot O(n^{3/2}\log^{3/2}n) = O(n^{3/2}\log^{5/2} n)$ bits.
\end{proof}

\paragraph{A Two-Phase Protocol for Perfect Matching.}
We now describe the protocol for \Cref{thm:matching}.
Fix an integer parameter $1 \le k < \frac{n}{2}$ to be optimized later.
The protocol proceeds in two phases.

\begin{itemize}
\item In Phase 1, Alice and Bob either find a matching of size $\frac{n}{2}-k$, or correctly conclude that $G$ has no perfect matching.
Phase 1 requires $O\!\left(\frac{n^2\log^2 n}{k}\right)$ bits of communication.
\item In Phase 2, starting from a matching of size $\frac{n}{2}-k$, Alice and Bob repeatedly find augmenting paths.
This requires $k$ iterations to reach a perfect matching if one exists.
By simulating the blossom algorithm~\cite{Edmonds65}, each augmenting path can be found using $O(n\log n)$ bits of communication.
Phase 2 requires $O(kn\log n)$ bits of communication.
\end{itemize}

Choosing $k=\Theta(n^{1/2} \log^{1/2} n)$ gives total communication
\[
O\!\left(\frac{n^2\log^2 n}{k}+kn\log n\right) = O(n^{3/2}\log^{3/2} n),
\]
which proves \Cref{thm:matching}.

We next describe Phase 1 and analyze its correctness and communication cost in
\Cref{subsec:general_matching_phase1}.
The implementation and analysis of
Phase 2 are deferred to \Cref{subsec:phase2}.

\subsection{Phase One: Eliminating Tutte-Berge Covers}
\label{subsec:general_matching_phase1}

In this section, we show how to implement Phase 1 of the protocol.

\begin{lemma}
\label{lem:PM_phase1_correctness}
Let $G = (V, E)$ be a general graph where $n = |V|$ is even.
Let $1 \le k < \frac{n}{2}$.
There is a deterministic protocol that, using $O\!\left(\frac{n^2\log^2 n}{k}\right)$ bits of communication, either finds a matching of $G$ of size $\frac{n}{2}-k$, or correctly concludes that $G$ has no perfect matching.
\end{lemma}

The key idea is based on the Tutte-Berge formula~\cite{Tutte47,Berge58} (\Cref{thm:tutte_berge_formula}), which gives a min-max characterization of the maximum matching size in general graphs.
Motivated by this, we define the notion of Tutte-Berge covers (\Cref{def:tutte-berge-cover}), which play a similar role as vertex covers in bipartite graphs: a low-value Tutte-Berge cover of $G$ certifies that $G$ cannot have a large matching.

Phase 1 maintains a public graph $H\subseteq G$ and the set $\mathcal{C}$ of low-value Tutte-Berge covers that cover $H$.
Alice and Bob repeatedly reveal edges to invalidate covers in $\mathcal{C}$.
If $\mathcal{C}$ becomes empty, then the Tutte-Berge formula implies that $H$ has a matching of the target size.
If some cover in $\mathcal{C}$ survives all edges of $G$, then it certifies that $G$ has no perfect matching.
The main point of the analysis is that, when $G$ has a perfect matching $M$, every remaining cover is invalidated by many edges of $M$, so the players can eliminate covers quickly.

\begin{definition}[Tutte-Berge Cover]
\label{def:tutte-berge-cover}
Let $V$ be a vertex set.
A Tutte-Berge cover is a tuple $C = (A, S_1, \ldots, S_t)$ such that
$A, S_1, \ldots, S_t$ form a partition of $V$, and each $S_i$ has odd size.

The value of $C$ is defined as
\[
\delta(C) := |A| + \sum_{i=1}^t \frac{|S_i|-1}{2}.
\]

We say that $C$ covers an edge $e = (u, v)$ if either $u \in A$ or $v \in A$,
or if both $u$ and $v$ are in $S_i$ for some $i \in [t]$.
For a graph $H$, we say that $C$ covers $H$ if it covers every edge in $E(H)$.
We say an edge $e$ invalidates $C$ if $C$ does not cover $e$.
\end{definition}

\begin{figure}[h]
\centering{
\begin{tikzpicture}[scale=0.9, auto, thick]
    \node[node, fill=blue!15] (v1) at (-3, 1.5) {$v_1$};
    \node[node] (v2) at (0, 1.5) {$v_2$};
    \node[node, fill=green!15] (v3) at (3, 1.5) {$v_3$};
    \node[node, fill=blue!15] (v4) at (-3, -1.5) {$v_4$};
    \node[node, fill=blue!15] (v5) at (0, -1.5) {$v_5$};
    \node[node, fill=red!15] (v6) at (3, -1.5) {$v_6$};

    \draw[edge] (v2) to (v1);
    \draw[edge] (v2) to (v3);
    \draw[edge] (v2) to (v4);
    \draw[edge] (v2) to (v5);
    \draw[edge] (v2) to (v6);

    \draw[edge] (v1) to (v4);
    \draw[edge] (v1) to (v5);
    \draw[edge] (v4) to (v5);

    \draw[dashed] (v3) to (v6);
\end{tikzpicture}
}
\caption{An example of a Tutte-Berge cover $C=(A=\{v_2\}, S_1=\{v_1,v_4,v_5\}, S_2=\{v_3\}, S_3=\{v_6\})$, whose value is $\delta(C)=|A|+\sum_i \frac{|S_i|-1}{2}=2$.
The solid edges are exactly the edges covered by $C$.
An edge with endpoints in two different odd sets invalidates $C$, e.g., the dashed
edge $(v_3,v_6)$.}
\label{fig:TB-cover}
\end{figure}
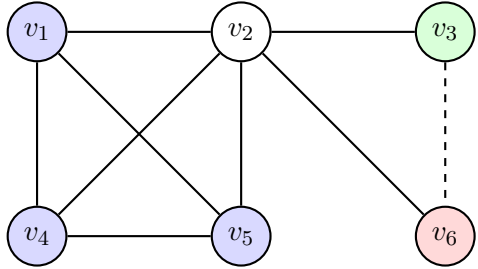

We use the following classical result as a black box.
\begin{theorem}[Tutte-Berge Formula~\cite{Tutte47,Berge58}]
\label{thm:tutte_berge_formula}
For every graph $G$,
\[
\max_{\textnormal{ matching $M$ of $G$}} |M|
=
\min_{\textnormal{ Tutte-Berge cover $C$ of $G$}} \delta(C).
\]
\end{theorem}

We now describe Phase 1 of the protocol (\Cref{alg:PM_phase1}).

\begin{algorithm}[H]
\caption{Phase 1 of $\mathcal{P}^{\mathsf{PM}}$}
\label{alg:PM_phase1}
\KwIn{A graph $G = (V, E)$ where $n = |V|$ is even, and an integer $1 \le k < \frac{n}{2}$.}
\KwOut{A matching of $G$ of size $\frac{n}{2}-k$, or conclude that $G$ has no perfect matching.}
$H \gets (V,\emptyset)$\;
$\mathcal{C} \gets \left\{
C : C \text{ is a Tutte-Berge cover with } \delta(C)= \frac{n}{2} - k - 1 \right\}$\;
$T \gets \bigl\lceil \frac{n^2\ln n}{k}\bigr\rceil$\;
\For{$i = 0, \ldots, T-1$}{
    Find an edge $e_i\in (G \setminus H)$ that invalidates the largest number
    of covers in $\mathcal{C}$\;
    \If{no edge invalidates any cover in $\mathcal{C}$}{
        \Return that $G$ has no perfect matching\; \label{line:phase1-no-edge}
    }
    $H = H \cup \{e_i\}$\;
    Remove from $\mathcal{C}$ every cover invalidated by $e_i$\;
    \If{$\mathcal{C}=\emptyset$}{
        \Return a matching of $H$ of size $\frac{n}{2}-k$\; \label{line:phase1-cover-empty}
    }
}
\Return that $G$ has no perfect matching\; \label{line:phase1-loop-exhausted}
\end{algorithm}

The following lemma will be useful for showing that the players can make good progress.

\begin{lemma}
\label{lem:tutte_berge_progress}
Let $C$ and $M$ be a Tutte-Berge cover and a matching on the same vertex set $V$, respectively.
Then $C$ covers at most $\delta(C)$ edges of $M$.

In particular, if $|V| = n$ is even, $M$ is a perfect matching, and $\delta(C) = \frac{n}{2} - k - 1$, then at least $k + 1$ edges of $M$ invalidate $C$.
\end{lemma}
\begin{proof}
Let $C = (A, S_1, \ldots, S_t)$.
By definition, $C$ covers an edge only if it is incident to $A$ or lies inside some $S_i$.
Because $M$ is a matching, at most $|A|$ edges of $M$ are incident to $A$, and at most $\frac{|S_i|-1}{2}$ edges lie inside each $S_i$ as $|S_i|$ is odd.
Therefore, $C$ can cover at most $|A| + \sum_{i=1}^t \frac{|S_i|-1}{2} = \delta(C)$ edges of $M$.

For the second claim, a perfect matching has $\frac{n}{2}$ edges, of which at most $\delta(C) = \frac{n}{2} - k - 1$ are covered.
The remaining $\frac{n}{2} - \bigl(\frac{n}{2} - k - 1\bigr) = k + 1$ edges are uncovered by, and each invalidates $C$.
\end{proof}

We are now ready to prove the correctness of \Cref{alg:PM_phase1} and analyze the communication cost of simulating it.

\begin{proof}[Proof of \Cref{lem:PM_phase1_correctness}]
Throughout \Cref{lem:PM_phase1_correctness}, $\mathcal{C}$ is maintained to be exactly the set of Tutte-Berge covers of value $\frac{n}{2}-k-1$ that cover the current graph $H$.
We consider each of the algorithm's return statements.

\textbf{Case 1}: The algorithm returns at Line~\ref{line:phase1-no-edge} in iteration $i$.
Then $\mathcal{C} \ne \emptyset$, otherwise the algorithm would have returned at Line~\ref{line:phase1-cover-empty} in the previous iteration.
Fix any $C \in \mathcal{C}$.
We know $C$ covers $H$, and since no edge in $G \setminus H$ invalidates $C$, $C$ covers $G$.
By \Cref{lem:tutte_berge_progress}, every matching of $G$ has size at most $\delta(C) = \frac{n}{2}-k-1 < \frac{n}{2}$.
The algorithm correctly decides that $G$ has no perfect matching.

\textbf{Case 2}: The algorithm returns at Line~\ref{line:phase1-cover-empty} in iteration $i$.
Then $H$ has no Tutte-Berge covers of value at most $\frac{n}{2}-k-1$.\footnote{For a Tutte-Berge cover $C$, we can increase $\delta(C)$ by $1$ without reducing its edge coverage, e.g., by moving a size-$1$ odd set into $A$ or moving two vertices from a larger odd set into $A$.}
Consequently, by \Cref{thm:tutte_berge_formula}, the maximum matching size of $H \subseteq G$ is at least $\frac{n}{2}-k$.
The algorithm correctly returns such a matching.

\textbf{Case 3}: The algorithm returns at Line~\ref{line:phase1-loop-exhausted} after $T$ iterations.
Suppose for contradiction that $G$ has a perfect matching $M$.

Let $\mathcal{C}_i$ denote $\mathcal{C}$ at the start of iteration $i$.
Then $\mathcal{C}_i \ne \emptyset$ for all $i = 0, \ldots, T-1$.
Fix any $i$.
By \Cref{lem:tutte_berge_progress}, each $C \in \mathcal{C}_i$ is invalidated by at least $k+1$ edges of $M$.
By double counting and averaging, some edge $e \in M$ invalidates at least a $\frac{2k}{n}$-fraction of the covers in $\mathcal{C}_i$.

The algorithm greedily chooses the next edge, which eliminates at least this many covers, so $|\mathcal{C}_{i+1}| \le \left(1-\frac{2k}{n}\right)|\mathcal{C}_i|$.
Since $|\mathcal{C}_0| \le (n+1)^n$, after $T$ iterations
\[
|\mathcal{C}_{T}| \le \left(1-\frac{2k}{n}\right)^T (n+1)^n \le \exp\!\left(-\frac{2kT}{n}\right)(n+1)^n < 1,
\]
using $1-x \le e^{-x}$ and our choice of $T = \bigl\lceil \frac{n^2\ln n}{k}\bigr\rceil$.
This contradicts $\mathcal{C}_T \ne \emptyset$.
Therefore, the algorithm correctly concludes that $G$ has no perfect matching.

Finally, we analyze the communication cost.
Since $H$ and $\mathcal{C}$ are public, each player can locally compute how many covers in $\mathcal{C}$ each of their own edge invalidates.
In each iteration, Alice sends her best edge and Bob sends his, and they take whichever invalidates more covers (breaking ties by taking Alice's edge).
Each iteration requires $O(\log n)$ bits of communication.
There are $T = O\!\left(\frac{n^2 \log n}{k}\right)$ iterations, so the total communication is $O(T \log n) = O\!\left(\frac{n^2 \log^2 n}{k}\right)$ bits.
\end{proof}

\section{A Cutting-Plane Framework for Polytope Emptiness Testing}
\label{sec:cutting_plane}

In this section, we describe a cutting-plane framework for testing whether a polytope is empty.
The core idea is closely related to classical optimization techniques such as the center-of-gravity and ellipsoid methods.
We do not claim novelty for this framework, as we distill it from~\cite{BlikstadBEMN22}, where it was used to obtain a near-linear communication protocol for bipartite matching.
By abstracting away the problem-specific aspects of bipartite matching, we present it as a general framework that can be applied to other communication problems.

Consider the following ``emptiness-testing'' problem: given a convex polytope $P \subseteq \mathbb{R}^n$ via a separation oracle, we are promised that either $P = \emptyset$ or $\vol(P) \ge \tau$ for some known volume threshold $\tau > 0$.
The goal is to determine which of the two cases holds using a small number of queries to the separation oracle of $P$.

\begin{algorithm}[ht]
\caption{A cutting-plane algorithm for polytope emptiness testing}
\label{alg:cutting_plane}
\KwIn{Access to a separation oracle for a convex polytope $P \subseteq \mathbb{R}^n$, a convex polytope $P_0 \supseteq P$ with known volume $\vol(P_0) > 0$, and a volume threshold $0 < \tau < \vol(P_0)$. It is promised that either $P = \emptyset$ or $\vol(P) \ge \tau$.}
\KwOut{Decide whether $P = \emptyset$ or $\vol(P) \ge \tau$. In the latter case, output a point $p \in P$.}
$T \gets \left\lceil 3 \ln(\vol(P_0)/\tau) \right\rceil$\;
\For{$i = 0, \ldots, T-1$}{
    $p_i \gets$ center of mass of $P_i$\;
    Query the separation oracle for $P$ at $p_i$: it either asserts that $p_i \in P$, or returns a linear constraint $\langle a^{(i)}, x\rangle \le b^{(i)}$ that is satisfied by every point of $P$ but violated by $p_i$\;
    \If{the oracle asserts that $p_i \in P$}{
        \Return that $\vol(P) \ge \tau$ and $p_i \in P$\; \label{line:feasible}
    }
    $P_{i+1} \gets P_{i} \cap \{x : \langle a^{(i)}, x\rangle \le b^{(i)}\}$\;
    \If{$P_{i+1} = \emptyset$}{
        \Return that $P = \emptyset$\; \label{line:infeasible-loop}
    }
}
\Return that $P = \emptyset$\; \label{line:infeasible-final}
\end{algorithm}

\begin{lemma}
\label{lem:cutting_plane_iterations}
Let $P_0 \subseteq \mathbb{R}^n$ be a convex polytope with volume $\vol(P_0) > 0$.
Let $0 < \tau < \vol(P_0)$.
Let $P \subseteq P_0$ be a convex polytope with the promise that either $P = \emptyset$ or $\vol(P) \ge \tau$.
Then \Cref{alg:cutting_plane} correctly distinguishes between the two cases after making at most $O\left(\log(\vol(P_0)/\tau)\right)$ queries to the separation oracle for $P$.
Moreover, in the case $\vol(P) \ge \tau$, \Cref{alg:cutting_plane} also returns a point $p \in P$.
\end{lemma}

We defer the proof of \Cref{lem:cutting_plane_iterations} to
\Cref{sec:cutting_plane_deferred_proofs}.
\Cref{lem:cutting_plane_iterations} suggests the following general recipe for solving a graph problem $\mathcal{T}$ via polytope emptiness testing:
\begin{enumerate}
    \item Reformulate a relevant (sub)problem as testing whether a convex polytope $P$ is empty or has volume at least $\tau$, for appropriate choices of $P$ and $\tau$.
    \item Show that a separation oracle for $P$ can be efficiently simulated in the communication model.
    \item Apply \Cref{alg:cutting_plane} to distinguish between $P = \emptyset$ and $\vol(P) \ge \tau$.
    \item Translate the outcome back to the original problem $\mathcal{T}$.
\end{enumerate}

This recipe is particularly well-suited for the two-party model.
Primal LP formulations of many graph problems (e.g., bipartite matching or shortest paths) have variables corresponding to edges, which are not jointly known by the players.
The key observation is that the dual LPs often have one \emph{constraint} per edge.
In many cases, each player can locally check the constraints corresponding to their own edges, and a violated constraint can be encoded by a short message.

As concrete examples, \cite{BlikstadBEMN22} applied this framework to a fractional vertex-cover polytope (the dual of fractional matching).
By K\H{o}nig's theorem, this polytope is empty if and only if the input graph has a perfect matching (see \Cref{subsec:BMM_cutting_plane}).
In \Cref{sec:negcycle-negSSSP}, we apply the same framework to a vertex-potential polytope (the dual of shortest paths) to obtain new communication protocols for negative cycle detection and negative-weight SSSP.

\section{Negative Cycle Detection and Negative-Weight Shortest Paths}
\label{sec:negcycle-negSSSP}

\subsection{Negative Cycle Detection}
\label{subsec:negcycle}
The problem of \emph{negative cycle detection} ($\negcycle$) is defined as follows: Alice and Bob are given a disjoint edge partition of a directed weighted graph $G = (V, E, w)$ with integer edge weights in $[-W, W]$. Their goal is to decide whether $G$ has a negative cycle.

Let $n = |V|$.
We consider the following vertex-potential polytope, whose constraints encode approximate triangle inequalities.
\begin{equation}
\label{eqn:triangle-ineq-polytope}
P^{\triangle}(G) := \left\{ x \in \mathbb{R}^n \; \middle| \;
\begin{aligned}
&x_v - x_u \le w(u, v) + \tfrac{1}{n^2} && \forall (u,v) \in E \\
&0 \le x_v \le nW && \forall v \in V
\end{aligned} \right\}.
\end{equation}

The additive $\frac{1}{n^2}$ slack will allow us to lower bound the volume of $P^{\triangle}(G)$ when it is nonempty.
At the same time, for integer
edge weights, this slack does not affect negative-cycle detection: we show that $P^{\triangle}(G)$ is empty if and only if $G$ has a negative cycle.

\begin{lemma}
\label{lem:P_neg_cycle}
Let $G = (V,E,w)$ be a directed weighted graph on $n \ge 2$ vertices with integer edge weights in $[-W, W]$.
Let $P^{\triangle}(G)$ be the polytope defined in \Cref{eqn:triangle-ineq-polytope}.
\begin{enumerate}
    \item If $G$ has a negative cycle, then $P^{\triangle}(G) = \emptyset$.
    \item If $G$ has no negative cycles, then $\vol\!\left(P^{\triangle}(G)\right) \ge n^{-2n}$.
\end{enumerate}
\end{lemma}
\begin{proof}
We first prove (1).
Suppose for contradiction that $P^{\triangle}(G) \ne \emptyset$, and let $x \in P^{\triangle}(G)$.
Let $(v_1,\ldots,v_k)$ be a negative cycle in $G$ where $k \le n$.
Then for each $i \in [k]$ (with $v_{k+1} = v_1$), we have $x_{v_{i+1}} - x_{v_i} \le w(v_i, v_{i+1}) + \frac{1}{n^2}$.
Since the cycle is negative and all edge weights are integers, $\sum_{i \in [k]} w(v_i, v_{i+1}) \le -1$.
Summing over $i \in [k]$, we have
$0 = \sum_{i \in [k]} (x_{v_{i+1}} - x_{v_i}) \le \sum_{i \in [k]} w(v_i, v_{i+1}) + \frac{k}{n^2} \le -1 + \frac{1}{n} < 0$, which is a contradiction.

We now prove (2).
When $G$ has no negative cycles, we construct a polytope $Q \subseteq P^{\triangle}(G)$ with volume $n^{-2n}$.
Let $G'$ be the graph obtained from $G$ by adding a new vertex $s'$ and directed edges $(s', v)$ of weight $(n-1)W$ for all $v \in V$.
Since $G'$ also has no negative cycles, shortest-path distances from $s'$ are well-defined.
Let $d(v) := d_{G'}(s',v)$ for each $v \in V$.

We note two properties of $d(\cdot)$.
First, $d(\cdot)$ satisfies the triangle inequality: $d(v) \le d(u) + w(u, v)$ for all $(u,v) \in E$.
Second, $d(v) \in [0, (n-1)W]$ for all $v \in V$.
The upper bound follows from the direct edge $(s',v)$, and for the lower bound, there is a shortest path $(s',v_1,\ldots,v_\ell = v)$ from $s'$ to $v$ with $\ell \le n$, so $d(v) = (n-1)W + \sum_{i=1}^{\ell-1} w(v_i, v_{i+1}) \ge (n-1)W + (\ell-1)(-W) \ge 0$.

Consider the $n$-dimensional cube $Q := \prod_{v \in V}\!\left[d(v), d(v) + \frac{1}{n^2}\right]$.
We can verify that $Q \subseteq P^{\triangle}(G)$.
Fix any $q = (q_v)_{v \in V} \in Q$.
We have $0 \le d(v) \le q_v \le d(v) + \frac{1}{n^2} \le (n-1)W + \frac{1}{n^2} \le nW$.
Moreover, for all $(u,v) \in E$, $q_v - q_u \le d(v) + \frac{1}{n^2} - d(u) \le w(u, v) + \frac{1}{n^2}$ by the triangle inequality for $d(\cdot)$.
Therefore, $\vol\!\left(P^{\triangle}(G)\right) \ge \vol(Q) = n^{-2n}$.
\end{proof}

We now present our protocol for $\negcycle$.
Note that when $G$ has no negative cycles, the protocol outputs a vector $p \in P^{\triangle}(G)$, which we will use in \Cref{subsec:negSSSP} to solve $\negSSSP$.

\begin{theorem}
\label{thm:P_neg_cycle}
There is a deterministic protocol $\mathcal{P}^{\negcycle}$ that solves $\negcycle$ using $O(n \log^2(nW))$ bits of communication.
Moreover, if the input graph $G$ has no negative cycles, the protocol outputs a vector $p \in P^{\triangle}(G)$, where $P^{\triangle}(G)$ is defined in \Cref{eqn:triangle-ineq-polytope}.
\end{theorem}
\begin{proof}
By \Cref{lem:P_neg_cycle}, it suffices to distinguish whether $P^{\triangle}(G) = \emptyset$ or $\vol\!\left(P^{\triangle}(G)\right) \ge n^{-2n}$.
Alice and Bob can do this by simulating \Cref{alg:cutting_plane} with $P = P^{\triangle}(G)$, $P_0 = [0,nW]^n$, and $\tau = n^{-2n}$.

Note that $P_0$ is public and independent of $G$.
Since $P_i$ is entirely determined by $P_0$ and previously communicated constraints, both players know $P_i$ and can locally compute the center of mass $p_i$ without communication.
The constraints $0 \le x_v \le nW$ are already captured by $P_0$.
To simulate the separation oracle for $P$ at $p_i$, Alice and Bob only need to check the edge constraints $x_v - x_u \le w(u, v) + \tfrac{1}{n^2}$, each of which can be checked locally by the player holding the edge $(u,v)$.
Both players send a violated constraint if they have one (or a \emph{null} message otherwise).
If both players send a violated constraint, they break ties by using Alice's.
If neither player sends a violated constraint, the oracle asserts $p_i \in P$.

By \Cref{lem:cutting_plane_iterations}, this protocol correctly distinguishes between the two promised cases after $O(\log(\vol(P_0)/\tau)) = O(\log(n^{3n}W^n)) = O(n \log(nW))$ queries to the separation oracle, and outputs a point $p \in P^{\triangle}(G)$ when $P^{\triangle}(G) \ne \emptyset$.
Each query to the separation oracle requires each player to communicate at most one constraint, and each constraint can be encoded using $O(\log(nW))$ bits (i.e., an edge $(u,v)$ and its weight).
Therefore, the total communication is $O(n \log^2(nW))$ bits.
\end{proof}

\subsection{Negative-Weight Single-Source Shortest Paths}
\label{subsec:negSSSP}

The problem of \emph{negative-weight single-source shortest paths} ($\negSSSP$) is defined as follows: Alice and Bob are given a disjoint edge partition of a directed weighted graph $G = (V, E, w)$ with integer edge weights in $[-W, W]$, and a source vertex $s\in V$.
Their goal is to either decide that $G$ has a negative cycle or compute the single-source shortest-path distances $d_G(s,v)$ for all $v\in V$.

We use a vertex-potential reweighting idea, familiar from Johnson's algorithm for all-pairs shortest paths~\cite{Johnson77}.
For a function $\phi: V \to \mathbb{R}$, let
\[
w_\phi(u,v) := w(u,v) + \phi(u) - \phi(v), \qquad \forall (u,v) \in E,
\]
and let $G_\phi := (V, E, w_\phi)$.
Along any path from $s$ to $v$, the $\phi(\cdot)$ terms telescope. Taking the minimum over all $s$-$v$ paths shows
\[
d_{G_\phi}(s,v) = d_G(s,v) + \phi(s) - \phi(v) \qquad \forall v \in V.
\]
In particular, if $w_\phi(e) \ge 0$ for every edge $e \in E$, one can compute $d_{G_\phi}(s, \cdot)$ using any nonnegative-weight shortest-path algorithm and recover $d_G(s,v)$.

In our application, setting $\phi$ to the vector $p \in P^{\triangle}(G)$ returned by \Cref{thm:P_neg_cycle} only guarantees $w_\phi(u, v) \ge -\frac{1}{n^2}$.
To obtain nonnegative edge weights, we round $w_\phi(u, v)$ up to the nearest multiple of $\frac{1}{n}$.
This rounding has only a mild effect on shortest-path distances: it increases the weight of any single edge by at most $\frac{1}{n}$ and the weight of any simple path by at most $\frac{n-1}{n} < 1$.
Since $d_G(s,v)$ is an integer, we can then run a nonnegative-weight
shortest-path algorithm and recover the original distances by taking floors at the end.\footnote{Analogous rounding arguments appear in the near-linear-time algorithm of~\cite{BernsteinNW22} for
negative-weight SSSP.}

\begin{lemma}[Reweighting with approximate potentials]
\label{lem:approx_johnson}
Let $G=(V,E,w)$ be a directed graph on $n\ge 2$ vertices with integer edge weights and no negative cycles.
Suppose $\phi:V\to\mathbb{R}$ satisfies $w(u,v) + \phi(u) - \phi(v) \ge -\frac{1}{n^2}$ for all $(u,v)\in E$.
Let $\widetilde G = (V, E, \widetilde w)$, where 
\[
\widetilde w(u,v) := \frac{1}{n} \Bigl\lceil n \bigl( w(u,v) + \phi(u) - \phi(v) \bigr) \Bigr\rceil \qquad \forall (u,v)\in E.
\]
Then, $\widetilde w(u,v)\ge 0$ for every $(u,v)\in E$.
Moreover, $G$ and $\widetilde G$ have the same reachability relation, and for
every vertex $v \in V$ reachable from $s$,
\[
d_G(s, v) = \Bigl\lfloor d_{\widetilde G}(s, v) + \phi(v) - \phi(s) \Bigr\rfloor.
\]
\end{lemma}
\begin{proof}
Let $w_\phi(u,v)=w(u,v)+\phi(u)-\phi(v)$ and $G_\phi = (V, E, w_\phi)$.
Since $w_\phi(u,v) \ge -\frac{1}{n^2}$, we have $n w_\phi(u,v) \ge -\frac{1}{n} > -1$.
Therefore, $\lceil n w_\phi(u,v)\rceil\ge 0$ and $\widetilde w(u,v)\ge 0$.

Because $\widetilde G$, $G_\phi$, and $G$ all share the same edge set $E$, they share the same reachability relation.

Fix any $v \in V$ reachable from $s$.
For every $e \in E$, $\widetilde w(e)$ rounds up $w_\phi(e)$ to the nearest multiple of $\frac{1}{n}$, so
\begin{equation}
\label{eqn:w_phi_w_tilde}
w_\phi(e) \le \widetilde w(e) < w_\phi(e) + \frac{1}{n}.
\end{equation}

Applying the left inequality of \Cref{eqn:w_phi_w_tilde} to a shortest $s$-$v$ path in $\widetilde G$ gives
\[
    d_{G_\phi}(s,v) \le d_{\widetilde G}(s,v).
\]
Since reweighting preserves cycle weights and $G$ has no negative cycles,
$G_\phi$ also has no negative cycles.
Therefore, there is a shortest simple $s$-$v$ path in $G_\phi$ with at most $n-1$ edges.
Applying the right inequality of \Cref{eqn:w_phi_w_tilde} to this path gives
\[
d_{\widetilde G}(s, v) < d_{G_\phi}(s, v) + \frac{n-1}{n} < d_{G_\phi}(s, v) + 1.
\]

Combining these inequalities with $d_{G_\phi}(s, v) = d_G(s, v) + \phi(s) - \phi(v)$, we have
\[
d_G(s, v) \le d_{\widetilde G}(s, v) + \phi(v) - \phi(s) < d_G(s, v)+1.
\]
Since $d_G(s,v)$ is an integer, taking the floor proves the lemma.
\end{proof}

We conclude this section with a deterministic protocol for $\negSSSP$.
\begin{theorem}
\label{thm:negSSSP}
There is a deterministic protocol $\mathcal{P}^{\negSSSP}$ that solves $\negSSSP$ using $O(n \log^2(nW))$ bits of communication.
\end{theorem}
\begin{proof}
The protocol proceeds as follows.
\begin{enumerate}
    \item Run $\mathcal{P}^{\negcycle}$ (\Cref{thm:P_neg_cycle}) on $G$. If a negative cycle is detected, terminate. Otherwise, both players obtain the same point $p \in P^\triangle(G)$.
    \item Let $\phi(v) := p_v$ for all $v \in V$.
          Let $\widetilde{G} = (V, E, \widetilde{w})$ where $\widetilde w(u,v) := \frac{1}{n} \left\lceil n \left( w(u,v) + \phi(u) - \phi(v) \right) \right\rceil$, as in \Cref{lem:approx_johnson}.
          Since both players know $p$, the edge weight $\widetilde{w}(u,v)$ can be computed locally by the player holding the edge $(u,v)$.
    \item Run $\mathcal{P}^{\mathsf{SSSP}}$ (\Cref{thm:SSSP}) to compute $d_{\widetilde{G}}(s,v)$ for all $v \in V$.
          Formally, $\mathcal{P}^{\mathsf{SSSP}}$ requires integer weights, so the players scale up $\widetilde w$ by a factor of $n$, run $\mathcal{P}^{\mathsf{SSSP}}$ on $(V, E, n \widetilde{w})$ to compute $n \cdot d_{\widetilde{G}}(s,v)$ for all $v \in V$, and then divide by $n$ to get $d_{\widetilde{G}}(s,v)$.
    \item Output $d_G(s,v) := \lfloor d_{\widetilde{G}}(s,v) + \phi(v) - \phi(s) \rfloor$ for every $v$ reachable from $s$, and $d_G(s,v) := \infty$ otherwise.
\end{enumerate}

For correctness, the constraints of $P^\triangle(G)$ as defined in \Cref{eqn:triangle-ineq-polytope} guarantee that $p_v - p_u \le w(u, v) + \frac{1}{n^2}$, or equivalently, $w_\phi(u, v) \ge -\frac{1}{n^2}$, so $\phi$ satisfies the prerequisite of \Cref{lem:approx_johnson}.
By \Cref{lem:approx_johnson}, the edge weights $\widetilde w(u, v)$ are nonnegative and multiples of $\frac{1}{n}$, which become nonnegative integers after scaling up by a factor of $n$.
By \Cref{thm:SSSP}, the players correctly compute $d_{\widetilde{G}}(s, v)$ via $\mathcal{P}^{\mathsf{SSSP}}$.
The correctness of $\mathcal{P}^{\negSSSP}$ then follows from \Cref{lem:approx_johnson}.

It remains to analyze the communication cost of the protocol.
In Step 1, $\mathcal{P}^{\negcycle}$ requires $O(n \log^2(nW))$ bits of communication by \Cref{thm:P_neg_cycle}.
In Step 3, since the constraints of $P^\triangle(G)$ guarantee that $0 \le p_v \le nW$ for all $v \in V$, we have $w_\phi(u,v) = w(u, v) + p_u - p_v \le (n+1)W$, and hence the scaled weights $n\widetilde{w}(u,v)$ are nonnegative integers that are $O(n^2 W)$, so $\mathcal{P}^{\mathsf{SSSP}}$ requires $O(n \log(nW))$ bits by \Cref{thm:SSSP}.
The total communication cost is $O(n \log^2(nW))$ bits.
\end{proof}

\section{Maximum Matching in Bipartite Graphs}
\label{sec:bipartite_matching}

In this section, we first restate the $O(n \log^2 n)$-bit protocol in \cite{BlikstadBEMN22} for bipartite matching, and then we present our combinatorial protocol (resulted from discretizing the said work) in \Cref{subsec:comb_BMM}.

\subsection{A Near-Linear Communication Protocol for Bipartite Matching}
\label{subsec:BMM_cutting_plane}

In the problem of \emph{bipartite perfect matching} ($\mathsf{BPM}$), Alice and Bob are given a disjoint edge partition of an undirected bipartite graph $G = (V, E)$. Their goal is to decide whether $G$ admits a perfect matching and output one if it does.
We assume without loss of generality that $n = |V|$ is even and that each side of $G$ has exactly $\frac{n}{2}$ vertices.
The breakthrough result of \cite{BlikstadBEMN22} gave a protocol for this problem with near-linear communication complexity.

\begin{theorem}[\cite{BlikstadBEMN22}]
There is a deterministic protocol $\mathcal{P}^{\mathsf{BPM}}$ which solves $\mathsf{BPM}$ using $O(n \log^2 n)$ bits of communication.
\label{thm:BPM}
\end{theorem} 

We briefly describe the protocol from \cite{BlikstadBEMN22}, before sketching a proof for \Cref{thm:BPM}.

Consider the following vertex-cover polytope of $G$:
\begin{equation}
\label{eqn:fractional_VC_polytope}
P^{\textnormal{VC}}(G) :=
\left\{ x \in \mathbb{R}^n \; \middle| \;
\begin{aligned}
\, & x_u + x_v \ge 1 && \forall (u, v) \in E, \\
\, & \textstyle\sum_{v\in V} x_v \le \frac{n}{2} - \frac{1}{2}, \\
\, & 0\le x_v \le 1 && \forall v \in V
\end{aligned}
\right\}.
\end{equation}

The following lemma is from~\cite{BlikstadBEMN22}.
For completeness, we provide its proof in \Cref{sec:bipartite_matching_proofs}.

\begin{restatable}[\cite{BlikstadBEMN22}]{lem}{volumegap}
Let $G$ be a bipartite graph with $\frac{n}{2}$ vertices on each side.
Let $P^{\textnormal{VC}}(G)$ be the polytope defined in \Cref{eqn:fractional_VC_polytope}.
\begin{enumerate}
    \item If $G$ has a perfect matching, then $P^{\textnormal{VC}}(G) = \emptyset$. 
    \item If $G$ has no perfect matching, then $\vol(P^{\textnormal{VC}}(G)) \ge (4n)^{-n}$. 
\end{enumerate}
\label{lem:P_vertex}
\end{restatable}

\begin{proof}[Proof sketch of \Cref{thm:BPM}]
By \Cref{lem:P_vertex}, it is sufficient to decide whether $P^{\textnormal{VC}}(G) = \emptyset$ or $\vol(P^{\textnormal{VC}}(G)) \ge (4n)^{-n}$.
Alice and Bob can do this by simulating \Cref{alg:cutting_plane} with $P = P^{\textnormal{VC}}(G)$, $P_0 = [0,1]^n \cap \{x: \sum_v x_v \le \frac{n}{2} - \frac{1}{2} \}$, and $\tau = (4n)^{-n}$.

Both players know $P_0$ and can compute the center of mass of each $P_i$ without communication.
To simulate the separation oracle for $P$, Alice and Bob only need to check the constraints $x_u + x_v \ge 1$, each of which can be checked locally by the player holding the edge $(u,v)$.
By \Cref{lem:cutting_plane_iterations}, this protocol correctly distinguishes the two promised cases after $O(\log(\vol(P_0)/\tau)) = O(n \log n)$ queries to the separation oracle.
Each query requires each player to communicate at most one constraint, and each constraint can be encoded using $O(\log n)$ bits (i.e., an edge $(u, v)$).
Therefore, the total communication is $O(n \log^2 n)$ bits.

Let $H$ be the graph formed by the communicated edges.
If \Cref{alg:cutting_plane} returns $P^{\textnormal{VC}}(G) = \emptyset$ because $\vol\!\left(P^{\textnormal{VC}}(H)\right) < (4n)^{-n}$, then by \Cref{lem:P_vertex}, $H$ must have a perfect matching.
\end{proof}

\subsection{A Combinatorial Protocol for Bipartite Perfect Matching}
\label{subsec:comb_BMM}
The cutting-plane framework (\Cref{alg:cutting_plane}) is inherently continuous: it operates over a polytope, tracks progress using its volume, and computes the center of mass to find a violated constraint that eliminates a constant fraction of the volume.
In this section, we present a combinatorial protocol for $\mathsf{BPM}$ that avoids these continuous operations.

\begin{theorem}
\label{thm:BPM_comb}
There is a deterministic protocol $\mathcal{P}^{\mathsf{BPM}}$ that does not rely on the cutting-plane framework (\Cref{alg:cutting_plane}) and solves $\mathsf{BPM}$ using $O(n \log^2 n)$ bits of communication.
\end{theorem}

Our idea is to discretize the vertex-cover polytope $P^{\textnormal{VC}}(G)$.
Consider a set of fractional vertex covers $Q^{\textnormal{VC}}(G)$ defined as follows:
\begin{equation}
\label{eqn:fractional_VC_lattice}
Q^{\textnormal{VC}}(G) :=
\left\{ x \in \mathbb{R}^n \; \middle| \;
\begin{aligned}
\, & x_u + x_v \ge 1 && \forall (u, v) \in E, \\
\, & \textstyle\sum_{v\in V} x_v \le \frac{n}{2} - \frac{1}{2}, \\
\, & x_v \in \left\{0, \tfrac{1}{d}, \ldots, 1\right\} && \forall v \in V
\end{aligned}
\right\}.
\end{equation}

We say that a nonnegative vector in $\mathbb{R}^n$ is \emph{$d$-uniform} if each of its coordinates is a multiple of $\frac{1}{d}$.
With this notation, $Q^{\textnormal{VC}}(G)$ is the set of $d$-uniform fractional vertex covers of $G$ of size at most $\frac{n}{2}-\frac{1}{2}$.
For sufficiently large $d = \mathrm{poly}(n)$, the discrete quantity $|Q^{\textnormal{VC}}(H)|$ is essentially proportional to $\vol(P^{\textnormal{VC}}(H))$ throughout the protocol.
As the players communicate edges and add them to a shared public graph, we can use this discrete quantity as a proxy for volume to measure progress.

We formally state our protocol for $\mathsf{BPM}$ in \Cref{alg:discrete_BMM}.

\begin{algorithm}[H]
\caption{A combinatorial protocol for bipartite matching}
\label{alg:discrete_BMM}
\KwIn{A bipartite graph $G = (V,E)$ with $n \ge 2$ vertices.}
\KwOut{Decide whether $G$ has a perfect matching, and output one if it does.}
$d \gets 8n^7$\;
$H_0 \gets (V,\emptyset)$\;
$T \gets \left\lceil 10 \, n \ln(d+1) \right\rceil$\;
\For{$i = 0, \ldots, T-1$}{
    Find an edge $e_i \in (G \setminus H_i)$ such that $|Q^{\textnormal{VC}}(H_{i} \cup \{e_i\})| \le 0.9 \cdot |Q^{\textnormal{VC}}(H_{i})|$, where $Q^{\textnormal{VC}}(H)$ is the set of $d$-uniform fractional vertex covers of $H$ as defined in \Cref{eqn:fractional_VC_lattice}\;
    \If{there is no such $e_i$}{
        \Return that $G$ has no perfect matching\; 
        \label{line:BMM-no-edge}
    }
    $H_{i+1} \gets H_{i} \cup \{e_i\}$\;
    \If{$Q^{\textnormal{VC}}(H_{i+1}) = \emptyset$}{
        \Return that $G$ has a perfect matching, together with a perfect matching of $H_{i+1}$\; \label{line:BMM-empty-cover}
    }
}
\end{algorithm}

\begin{restatable}[Existence of a progress-making edge]{lem}{discreteBMMprogress}
\label{lem:discrete_BMM_progress}
Let $G$ be a bipartite graph with a perfect matching and $\frac{n}{2}$ vertices on each side.
Let $H$ be any subgraph of $G$ on the same vertices with $Q^{\textnormal{VC}}(H) \ne \emptyset$, where $Q^{\textnormal{VC}}(H)$ is defined in \Cref{eqn:fractional_VC_lattice} for $d=8n^7$.
Then, there is an edge $e \in (G \setminus H)$ such that
\[
\left|Q^{\textnormal{VC}}(H \cup \{e\})\right| \le 0.9 \cdot \left|Q^{\textnormal{VC}}(H)\right|.
\]
\end{restatable}

We defer the proof of \Cref{lem:discrete_BMM_progress} to \Cref{sec:discrete_BMM_progress_proofs} and first use it to prove \Cref{thm:BPM_comb}.

\begin{proof}[Proof of \Cref{thm:BPM_comb}]
We first show that \Cref{alg:discrete_BMM} always returns an answer.
Suppose for contradiction that it finishes all $T$ iterations without reaching Line~\ref{line:BMM-no-edge} or Line~\ref{line:BMM-empty-cover}.
Then
\[
\left|Q^{\textnormal{VC}}(H_T)\right| \le (0.9)^T \cdot \left|Q^{\textnormal{VC}}(H_0)\right| \le (0.9)^T \cdot (d+1)^n < 1,
\]
where the last inequality follows from our choice of $T = \lceil 10\, n \ln(d+1) \rceil$.
This implies $\left|Q^{\textnormal{VC}}(H_T)\right| = 0$, but then the algorithm would have returned in Line~\ref{line:BMM-empty-cover} of the final iteration, which is a contradiction.
We proceed by verifying correctness at the algorithm's return statements.

\textbf{Case 1}: \Cref{alg:discrete_BMM} returns at Line~\ref{line:BMM-no-edge} in iteration $i$.
Then $Q^{\textnormal{VC}}(H_i) \ne \emptyset$: this is true for $i = 0$, and for $i > 0$ the algorithm would have returned at Line~\ref{line:BMM-empty-cover} in the previous iteration.
If $G$ had a perfect matching, \Cref{lem:discrete_BMM_progress} would guarantee the existence of an edge $e_i$ satisfying the required condition.
Hence $G$ has no perfect matching, and the algorithm correctly decides this.

\textbf{Case 2}: \Cref{alg:discrete_BMM} returns at Line~\ref{line:BMM-empty-cover} in iteration $i$.
Then $Q^{\textnormal{VC}}(H_{i+1}) = \emptyset$.
Observe that any integral vertex cover of $H_{i+1}$ of size $\frac{n}{2} - 1$ is a feasible point in $Q^{\textnormal{VC}}(H_{i+1})$.
Therefore, $H_{i+1}$ has no vertex cover of size $\frac{n}{2} - 1$.
By K\H{o}nig's theorem, $H_{i+1}$ has a perfect matching.
Because $H_{i+1} \subseteq G$, this is also a perfect matching of $G$, which the algorithm outputs.

We now analyze the communication cost of simulating \Cref{alg:discrete_BMM} in the two-party setting.
Note that each $H_i$ is public.
In each iteration $i$, both players locally check if they have an edge $e \notin H_i$ satisfying $\left|Q^{\textnormal{VC}}(H_i \cup \{e\})\right| \le 0.9 \cdot \left|Q^{\textnormal{VC}}(H_i)\right|$, and send such an edge if they have one.
If both players send an edge, they break ties by using Alice's.
If neither player sends an edge, the algorithm returns at Line~\ref{line:BMM-no-edge}.
The agreed-upon edge is then added to form $H_{i+1}$, and both players can locally check whether $Q^{\textnormal{VC}}(H_{i+1}) = \emptyset$.
Therefore, each iteration requires each player to communicate at most one edge, which can be encoded using $O(\log n)$ bits.
Since $d = 8n^7$ and $T = \left\lceil 10\, n \ln(d+1) \right\rceil = O(n \log n)$, the total communication is $O(n \log^2 n)$ bits.
\end{proof}

\printbibliography

\appendix
\section{Finding Augmenting Paths via the Blossom Algorithm}
\label{subsec:phase2}
Fix an arbitrary matching $M \subseteq E$. An \emph{augmenting path} is a path in $G$ that begins and ends at distinct unmatched vertices, and alternates between edges not in $M$ and edges in $M$. By flipping the matching status of each edge along such a path, we increase the size of the matching by one. The following fact about augmenting paths is standard.
\begin{fact}
\label{fact:augmenting_path_max}
$G$ contains an augmenting path if and only if $M$ is not a maximum matching.
\end{fact}

In this section, we describe a procedure for finding an augmenting path in general graphs using $O(n \log n)$ bits. Hence, by repeating the process (up to) $k$ times starting from a matching of size $\frac{n}{2}-k$ obtained from Phase 1, one can obtain a maximum matching.

To identify an augmenting path, we follow Edmonds' Blossom Algorithm~\cite{Edmonds65} by constructing a \emph{maximal alternating forest}. Initially, the forest is empty, and all vertices are marked as ``U'' (unvisited). The forest grows from the unmatched vertices as roots by alternately exploring unmatched and matched edges. Each vertex, once added to the forest, is dynamically re-labeled based on the parity of its distance from its root: vertices at even distances are marked ``E'', and those at odd distances are marked ``O''. For convenience, we refer to an edge $(a, b)$ as an \emph{E-O edge} if $a$ is marked ``E'' and $b$ is marked ``O''. This notation extends analogously to other type combinations. The primary goal is to incrementally grow a maximal alternating forest and assign parity labels to vertices such that no E-U edge remains in the graph (hence the term ``maximal''). An example of such a maximal alternating forest is illustrated in \Cref{fig:maximal_alt_forest}.

\begin{enumerate}

\item[(1)] Mark every unmatched vertex as ``E''; these vertices are also the roots of the alternating forest. Mark all remaining vertices as ``U''. Let $F$ denote an edge set that is initially empty.

\item[(2)] Repeat the following steps until no E-U edge remains in $G$:
\begin{enumerate}
    \item[(2a)] Pick any E-U edge $(a, b) \in G \setminus F$.
    
    \item[(2b)] Add edge $(a, b)$ to $F$ and label vertex $b$ as ``O''.  
    
    \item[(2c)] Let $b'$ be the unique neighbor of $b$ such that $(b, b') \in M$. Note that such a $b'$ must exist, as otherwise $b$ would have been unmatched and thus labeled ``E'' in Step 1.
    
    \item[(2d)] Add edge $(b, b')$ to $F$ and label vertex $b'$ as ``E''.
\end{enumerate}

\item[(3)] \hyperlink{step3}{To be continued.} 
\end{enumerate}



\begin{figure}[ht]
\begin{center}
\begin{tikzpicture}[scale=0.7, thick
  node/.style={circle, draw, minimum size=8mm},
  every edge/.style={draw=blue, thick},
   E/.style={circle, draw, minimum size=5mm, fill=cyan!30},
    O/.style={circle, draw, minimum size=5mm, fill=orange!30},
    U/.style={circle, draw, minimum size=5mm, fill=green!30},
  reddash/.style={draw=red, dashed}
]


\node[E] (E1) at (0,0) {};
\node[O] (O1) at (0,-1.5) {};
\node[E] (E2) at (0,-3) {};
\node[O] (O2) at (-1.1,-4.5) {};
\node[O] (O3) at (1.1,-4.5) {};
\node[E] (E3) at (-1.1,-6) {};
\node[E] (E4) at (1.1,-6) {};
\node[U] (U1) at (1.4,-1.5) {};

\draw (E1) -- (O1);
\draw[ultra thick] (O1) -- (E2);
\draw (E2) -- (O2);
\draw[ultra thick] (O2) -- (E3);
\draw (E2) -- (O3);
\draw[ultra thick] (O3) -- (E4);
\draw[dashed] (O1) -- (U1);

\node[E] (E5) at (4,0) {};
\node[O] (O4) at (2.9,-1.5) {};
\node[O] (O5) at (5.1,-1.5) {};
\node[E] (E6) at (2.9,-3) {};
\node[E] (E7) at (5.1,-3) {};
\node[O] (O6) at (2.9,-4.5) {};
\node[O] (O7) at (4,-4.5) {};
\node[O] (O8) at (6.2,-4.5) {};
\node[E] (E8) at (4,-6) {};
\node[E] (E9) at (6.2,-6) {};
\node[O] (O9) at (3,-7.5) {};
\node[O] (O10) at (5.1,-7.5) {};
\node[O] (O11) at (6.2,-7.5) {};
\node[E] (E10) at (3,-9) {};
\node[E] (E11) at (5.1,-9) {};
\node[E] (E12) at (6.2,-9) {};
\node[U] (U2) at (6.5,-3) {};

\draw (E5) -- (O4);
\draw[ultra thick] (O4) -- (E6);
\draw(E6) -- (O6);
\draw (E5) -- (O5);
\draw[ultra thick] (O5) -- (E7);
\draw (E7) --(O7);
\draw[ultra thick] (O7) -- (E8);
\draw (E7) -- (O8);
\draw[ultra thick] (O8) -- (E9);
\draw (E9) -- (O11);
\draw[ultra thick] (O11) -- (E12);
\draw (E8) --(O9);
\draw[ultra thick] (O9) -- (E10);
\draw (E8) --(O10);
\draw[ultra thick] (O10) -- (E11);
\draw[dashed] (E4) -- (E6);
\draw[dashed] (E8) -- (E9);
\draw[dashed] (O8) -- (U2);

\node[above] at ($(E4)!0.5!(E6)$) {$e_1 \ $};
\node[above] at ($(E8)!0.5!(E9)$) {$e_2$};

\end{tikzpicture}
\caption{An example of a maximal alternating forest resulting from Steps 1-3. The blue vertices represent even vertices (E), with the top two being the roots of their respective trees. The orange vertices represent odd vertices (O), while the green vertices are unmarked (U). Solid edges are those included in the forest, with thick edges specifically indicating edges in the matching. Dashed edges denote edges in the graph that are not included in the forest.}
\label{fig:maximal_alt_forest}
\end{center}
\end{figure}
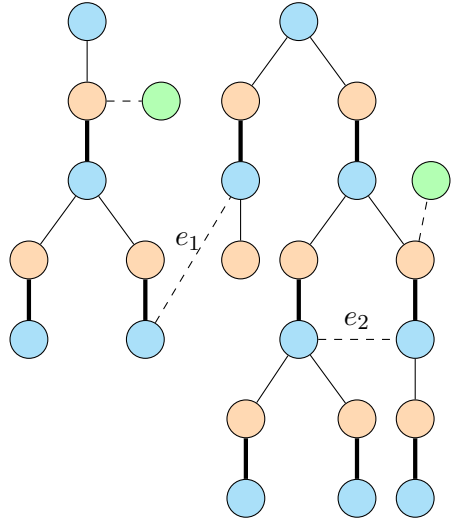

Such a process can be efficiently simulated in the setting of two-party communication. It requires only one of the players to locally identify their E-U edge and communicate it to the other player. The remaining steps of the procedure can be executed jointly.

Following this procedure, some vertices are re-labeled according to their parity in the alternating forest $F$. At the end of this phase, no E-U edges remain in the graph $G$. The following fact captures an important structural consequence of this property.

\begin{fact}
Suppose that $G$ contains an augmenting path and that no E-U edges remain in $G$. Then, there must exist an E-E edge in $G \setminus F$.
\label{fact:E-E_exists}
\end{fact}

\begin{proof}
Observe that there is no E-E edge in $F$. Assume for contradiction that every edge incident to a vertex labeled ``E'' has its other endpoint labeled ``O''. Furthermore, by construction, every matching edge in $M$ must be either a U-U edge or an E-O edge. Now consider an augmenting path $(a_1, b_1, a_2, b_2, \ldots, a_\ell, b_\ell)$, where $a_1$ and $b_\ell$ are unmatched vertices which have label ``E'' via Step 1. If $\ell=1$, then $(a_1,b_1)$ is already an E-E edge, a contradiction. Then assume $\ell\geq 2$. By the contradiction assumption, $b_1$ is labeled ``O''. Since $(b_1, a_2)$ is a matching edge, $a_2$ is labeled ``E''. Continuing this inductively, $b_\ell$ must be labeled ``O'', which leads to a contradiction.
\end{proof}

This fact allows us to continue the algorithm. If we are able to find an E-E edge in the graph, we proceed by adding it to the alternating forest. One of two outcomes will then occur (demonstrated in \Cref{fig:added_E-E}):
\begin{itemize}
    \item We discover an augmenting path, in which case the algorithm terminates successfully.
    \item We encounter a structure known as a \emph{blossom}---an odd-length cycle that includes exactly one vertex, called the \emph{tip}, which is incident to two non-matching edges, while the remaining edges of the cycle alternate between matched and unmatched edges.
\end{itemize}

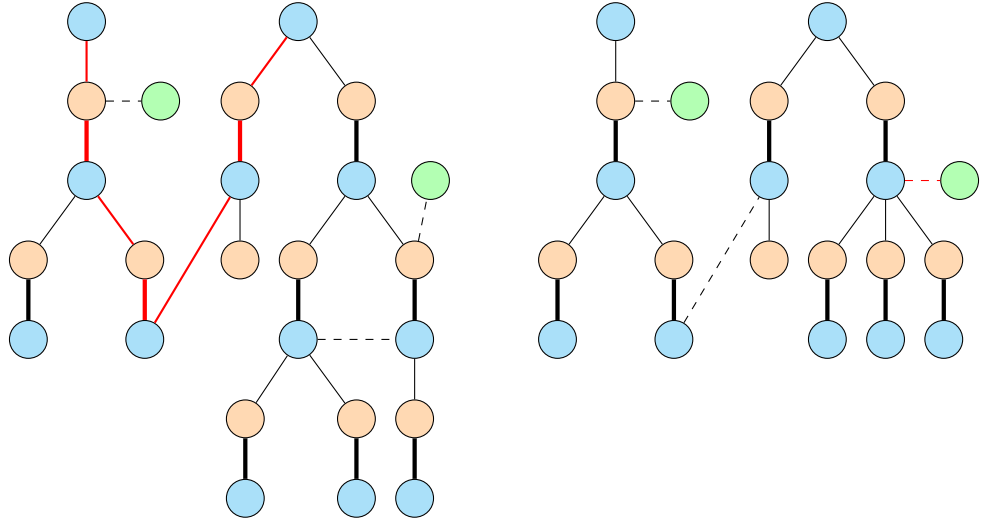
\begin{figure}[ht]
\begin{center}
\begin{tikzpicture}[scale=0.7, thick
  node/.style={circle, draw, minimum size=8mm},
  every edge/.style={draw=blue, thick},
   E/.style={circle, draw, minimum size=5mm, fill=cyan!30},
    O/.style={circle, draw, minimum size=5mm, fill=orange!30},
    U/.style={circle, draw, minimum size=5mm, fill=green!30},
  rededge/.style={draw=red, thick}
]


\node[E] (E1) at (0,0) {};
\node[O] (O1) at (0,-1.5) {};
\node[E] (E2) at (0,-3) {};
\node[O] (O2) at (-1.1,-4.5) {};
\node[O] (O3) at (1.1,-4.5) {};
\node[E] (E3) at (-1.1,-6) {};
\node[E] (E4) at (1.1,-6) {};
\node[U] (U1) at (1.4,-1.5) {};

\node[E] (E5) at (4,0) {};
\node[O] (O4) at (2.9,-1.5) {};
\node[O] (O5) at (5.1,-1.5) {};
\node[E] (E6) at (2.9,-3) {};
\node[E] (E7) at (5.1,-3) {};
\node[O] (O6) at (2.9,-4.5) {};
\node[O] (O7) at (4,-4.5) {};
\node[O] (O8) at (6.2,-4.5) {};
\node[E] (E8) at (4,-6) {};
\node[E] (E9) at (6.2,-6) {};
\node[O] (O9) at (3,-7.5) {};
\node[O] (O10) at (5.1,-7.5) {};
\node[O] (O11) at (6.2,-7.5) {};
\node[E] (E10) at (3,-9) {};
\node[E] (E11) at (5.1,-9) {};
\node[E] (E12) at (6.2,-9) {};
\node[U] (U2) at (6.5,-3) {};

\draw[red, thick] (E1) -- (O1);
\draw[ultra thick, red] (O1) -- (E2);
\draw (E2) -- (O2);
\draw[ultra thick] (O2) -- (E3);
\draw[red, thick] (E2) -- (O3);
\draw[ultra thick, red] (O3) -- (E4);
\draw[dashed] (O1) -- (U1);

\draw[thick, red] (E5) -- (O4);
\draw[ultra thick, red] (O4) -- (E6);
\draw(E6) -- (O6);
\draw (E5) -- (O5);
\draw[ultra thick] (O5) -- (E7);
\draw (E7) --(O7);
\draw[ultra thick] (O7) -- (E8);
\draw (E7) -- (O8);
\draw[ultra thick] (O8) -- (E9);
\draw (E9) -- (O11);
\draw[ultra thick] (O11) -- (E12);
\draw (E8) --(O9);
\draw[ultra thick] (O9) -- (E10);
\draw (E8) --(O10);
\draw[ultra thick] (O10) -- (E11);
\draw[thick, red] (E4) -- (E6);
\draw[dashed] (E8) -- (E9);
\draw[dashed] (O8) -- (U2);

\node[E] (E1') at (10,0) {};
\node[O] (O1') at (10,-1.5) {};
\node[E] (E2') at (10,-3) {};
\node[O] (O2') at (8.9,-4.5) {};
\node[O] (O3') at (11.1,-4.5) {};
\node[E] (E3') at (8.9,-6) {};
\node[E] (E4') at (11.1,-6) {};
\node[U] (U1') at (11.4,-1.5) {};

\draw (E1') -- (O1');
\draw[ultra thick] (O1') -- (E2');
\draw (E2') -- (O2');
\draw[ultra thick] (O2') -- (E3');
\draw (E2') -- (O3');
\draw[ultra thick] (O3') -- (E4');
\draw[dashed] (O1') -- (U1');

\node[E] (E5') at (14,0) {};
\node[O] (O4') at (12.9,-1.5) {};
\node[O] (O5') at (15.1,-1.5) {};
\node[E] (E6') at (12.9,-3) {};
\node[E] (E7') at (15.1,-3) {};
\node[O] (O6') at (12.9,-4.5) {};
\node[O] (O7') at (14,-4.5) {};
\node[O] (O8') at (15.1,-4.5) {};
\node[O] (O9') at (16.2,-4.5) {};
\node[E] (E8') at (14,-6) {};
\node[E] (E9') at (15.1,-6) {};
\node[E] (E10') at (16.2,-6) {};
\node[U] (U2') at (16.5,-3) {};

\draw (E5') -- (O4');
\draw[ultra thick] (O4') -- (E6');
\draw (E6') -- (O6');
\draw (E5') -- (O5');
\draw[ultra thick] (O5')-- (E7');
\draw (E7') --(O7');
\draw[ultra thick] (O7') -- (E8');
\draw (E7') --(O8');
\draw[ultra thick] (O8') -- (E9');
\draw (E7') --(O9');
\draw[ultra thick] (O9')-- (E10');
\draw[dashed] (E4') -- (E6');
\draw[dashed, red] (E7') -- (U2');

\end{tikzpicture}
\caption{An illustration of two possible outcomes when an E-E edge is added to a maximal alternating forest, exemplified by edges $e_1$ and $e_2$ from \Cref{fig:maximal_alt_forest}. On the left, the two even vertices of $e_1$ have different roots. In this case, the newly added edge completes an augmenting path, highlighted in red. On the right, the two even vertices of $e_2$ share the same root. In this case, the newly added edge completes a blossom. We then contract this blossom into a single even super-vertex, positioned at its tip. This contraction may introduce a new E-U edge (shown as a red dashed edge), which in turn triggers the algorithm to recurse to Step 2 to continue growing a maximal alternating forest.}
\label{fig:added_E-E}
\end{center}
\end{figure}

We emphasize the latter scenario. Blossoms are useful objects because they preserve the existence of augmenting paths even when contracted into a single supervertex. The following fact formalizes this intuition:

\begin{fact}[\cite{Edmonds65}]
\label{fact:blossom_preservation}
Let $B$ be a blossom in $G$ that is discovered by the addition of an E-E edge to a maximal alternating forest. Then $G$ contains an augmenting path (w.r.t. $M$) if and only if the graph obtained by contracting $B$ into a single supervertex also contains an augmenting path (wrt. the induced matching in the contracted graph).
\end{fact}

With this fact in mind, the algorithm proceeds to:

\begin{enumerate}
\setcounter{enumi}{2}
\item [\hypertarget{step3}{(3)}] Attempt to find an E-E edge $(a,b) \in G \setminus F$.
\begin{enumerate}
\item[(3a)] If no such E-E edge exists, terminate and confirm that $M$ is a maximum matching.
\item[(3b)] If $(a,b)$ is an E-E edge and $a$ and $b$ lie in alternating trees with different roots, then an augmenting path is found. Terminate and augment $M$.
\item[(3c)] If $(a,b)$ is an E-E edge and both vertices share the same tree root, then a blossom is found. Contract this blossom into a single supervertex (with the tip as its representative), and assign it the label ``E''. Note that this contraction may introduce new E-U edges. Recurse to Step (2) to continue growing the maximal alternating forest on the contracted graph.
\end{enumerate}
\end{enumerate}

Similar to before, this process only requires the players to locally identify an E-E edge and communicate it to the other player. Then, upon terminating via Step (3b), among the edges revealed so far, there must exist an augmenting path wrt. $M$, and both players can locally find a consistent augmenting path and augment $M$.

\subsection{Analysis}
Correctness follows from the standard analysis of the blossom algorithm.
We focus on bounding the communication cost.

We claim that the following potential function never increases during the execution of the algorithm before reaching Step (3a) or Step (3b) after which it terminates.
\begin{align*}\Phi \ := & \ \underbrace{(\text{\# edges communicated})}_{\Phi_1} \ + \ \underbrace{(\text{\# unmarked vertices})}_{\Phi_2} \ + \ \underbrace{(n - \text{\# blossoms contracted})}_{\Phi_3}.
\end{align*}

To justify this, we examine the two key steps where the potential may change:

\begin{itemize}
    \item In Step 2, whenever an E-U edge is revealed by one player, two new vertices are marked. Hence, $\Phi_1$ increases by $1$ and  $\Phi_2$ decreases by $2$, so $\Phi$ decreases by $1$.
    
    \item In Step 3c, when an E-E edge with the same root is found, this completes a blossom. One edge is communicated to complete the cycle, increasing $\Phi_1$ by 1. The number of blossoms also increases by 1, reducing $\Phi_3$ by 1. Thus, the net change in $\Phi$ is zero.
\end{itemize}

Initially, no edges have been communicated, at most $n$ vertices are unmarked, and no blossoms have been contracted. Thus, the starting potential is $\Phi \leq 2n$ and afterwards it only increases by one if the algorithm terminates via Step (3b). At termination, the number of edges communicated is exactly $\Phi_1$. Moreover, the number of unmarked vertices is nonnegative (i.e., $\Phi_2 \geq 0$), and the number of blossoms is at most $n$ (i.e., $\Phi_3 \geq 0$), as each blossom reduces the number of supervertices by at least one. Thus, the number of edges communicated is bounded by $\Phi_1 \leq \Phi \leq 2n+1$. Since each communicated edge requires $O(\log n)$ bits to describe, this implies the procedure of finding an augmenting path can be simulated using $O(n \log n)$ bits of communication between Alice and Bob.
\section{Omitted Proofs from Section~\ref{sec:cutting_plane}}
\label{sec:cutting_plane_deferred_proofs}

\begin{proof}[Proof of \Cref{lem:cutting_plane_iterations}]
\Cref{alg:cutting_plane} runs for at most $T$ iterations and queries the separation oracle exactly once per iteration, so it makes at most $T = \left\lceil 3 \ln(\vol(P_0)/\tau) \right\rceil = O(\log(\vol(P_0)/\tau))$ queries in total.
It remains to show that its output is always correct.

Note that $P \subseteq P_i$ for every $0 \le i \le T$.
This holds initially by the assumption $P \subseteq P_0$, and continues to hold because $P_{i+1}$ is formed by intersecting $P_i$ with a halfspace that contains every point of $P$.
We proceed by verifying correctness at each of the algorithm's three return statements.

\textbf{Case 1}: \Cref{alg:cutting_plane} returns at Line~\ref{line:feasible} when the separation oracle asserts that $p_i \in P$.
Then $P \ne \emptyset$, and by the promise on $P$, the algorithm correctly decides $\vol(P) \ge \tau$ and returns a point $p_i \in P$.

\textbf{Case 2}: \Cref{alg:cutting_plane} returns at Line~\ref{line:infeasible-loop} in iteration $i$.
Then $P \subseteq P_{i+1} = \emptyset$, and the algorithm correctly decides $P = \emptyset$.

\textbf{Case 3}: \Cref{alg:cutting_plane} returns at Line~\ref{line:infeasible-final} after $T$ iterations.
Then $P_T \ne \emptyset$.
Because
\[
P_{i+1} = P_i \cap \{x : \langle a^{(i)}, x \rangle \le b^{(i)}\}
\]
and the center of mass $p_i$ of $P_i$ violates this constraint, by Gr\"unbaum's Theorem \cite{Grunbaum1960}\footnote{Specifically, Gr\"unbaum~\cite{Grunbaum1960} showed that for any convex body $K$ and any halfspace $H$ that does not contain the center of mass of $K$, $\vol(K \cap H) \le \left(1 - \frac{1}{e}\right) \vol(K)$.}, we have
\[
\vol(P_{i+1}) \le \left(1 - \frac{1}{e}\right) \vol(P_i).
\]
By iteratively applying this inequality, we get
\[
\vol(P_T) \le \left(1 - \frac{1}{e}\right)^T \vol(P_0) \le e^{-T/e} \cdot \vol(P_0) < \tau,
\]
where we used $1 - x \le e^{-x}$, $\vol(P_0)/\tau > 1$, and $T \ge 3 \ln(\vol(P_0)/\tau) > e \ln(\vol(P_0)/\tau)$.
Because $P \subseteq P_T$, we have $\vol(P) \le \vol(P_T) < \tau$.
By the promise on $P$, we must have $P = \emptyset$, and the algorithm correctly decides that.
\end{proof}
\section{Directed Cycle Detection}
\label{subsec:kosaraju}

The problem of \emph{directed cycle detection} ($\cycle$) is defined as follows: the players are given an edge partition of a directed graph $G = (V, E)$ without self-loops and with $n = |V|$. Their goal is to determine whether $G$ contains a directed cycle.

\begin{theorem}
\label{thm:P_cycle}
There is a deterministic protocol $\mathcal{P}^{\mathsf{Cycle}}$ which solves $\cycle$ using $O(n \log n)$ bits of communication.
\end{theorem}

We simulate Kosaraju's algorithm to compute all strongly connected components (SCCs) of $G$. Recall that Kosaraju's algorithm performs two depth-first searches (DFS). We first discuss how to simulate DFS on a directed graph with a specific vertex ordering (\Cref{alg:dfs}).

\SetKwData{Visited}{Visited}
\SetKwFunction{Search}{Search}
\begin{algorithm}[ht]
\caption{Depth-First Search (DFS)}
\label{alg:dfs}
\KwIn{A directed graph $G = (V,E)$ and an ordering $L = (v_1, \ldots, v_n)$ of $V$.}
\KwOut{Runs a DFS on $G$ starting from vertices in the order $L$.}
\SetKwProg{Fn}{Function}{}{}
\Fn{\Search{$v$}}{
    $\Visited[v] \gets \mathrm{true}$\;
    \For{each edge $(v,u) \in E$}{ \label{line:dfs-check-vu}
        \If{$\Visited[u] = \mathrm{false}$}{
            \Search{$u$}\; \label{line:dfs_inner_call}
        }
    }
    \Return\; \label{line:dfs_return}
}
$\Visited[v] \gets \mathrm{false}$ for all $v \in V$\;
\For{$i = 1, \ldots, n$}{
    \If{$\Visited[v_i] = \mathrm{false}$}{
        \Search{$v_i$}\; \label{line:dfs_outer_call}
    }
}
\end{algorithm}

In Line~\ref{line:dfs-check-vu}, we simulate the loop over $v$'s edges as follows: Alice repeatedly checks her remaining (unchecked) edges $(v,u)$ for one with $u$ unvisited; each time she finds one, she sends it to Bob, and both proceed to the recursive $\Search$ call in Line~\ref{line:dfs_inner_call} before returning to check her remaining edges.
Once she has exhausted all her edges leaving $v$, she sends a single ``done'' message, and then it is Bob's turn to do the same with his edges.
Once Bob exhausts his edges and sends ``done'', they backtrack at Line~\ref{line:dfs_return}.
Both players maintain a correct copy of the $\Visited$ array.
For each vertex, each player sends exactly one ``done'' message.
Each communicated edge leads to an unvisited vertex, which is immediately marked as visited, so at most $n-1$ edges are communicated in total.
Each edge requires $O(\log n)$ bits to encode, so the total communication for one DFS is $O(n \log n)$ bits.

We now build on this to discuss how to simulate Kosaraju's algorithm.

\begin{proof}[Proof of \Cref{thm:P_cycle}]
We simulate Kosaraju's algorithm:
Perform a first DFS on $G$ with an arbitrary vertex ordering, and push a vertex $v$ onto a stack immediately before $\Search(v)$ returns at Line~\ref{line:dfs_return}.
Then perform a second DFS on the transpose graph $G^\top$ (where the direction of each edge is reversed, and each player can locally reverse their edges) using the stack order (last-in-first-out, i.e., in decreasing order of their finish times from the first DFS) as the vertex ordering $L$.
All vertices visited within a single outer call to $\Search$ in Line~\ref{line:dfs_outer_call} in the second DFS form one SCC.

Note that $G$ has a directed cycle if and only if there is at least one SCC of size at least $2$.
Therefore, after computing all SCCs, Alice and Bob can decide $\cycle$ without further communication.
In total, the protocol simulates two DFS traversals and uses $O(n\log n)$ bits of communication.
\end{proof}
\section{Nonnegative-Weight Single-Source Shortest Paths}
\label{subsec:dijkstra}

The problem of \emph{single-source shortest paths} ($\SSSP$) is as follows: the players are given a disjoint edge partition of a directed weighted graph $G = (V, E, w)$ whose weights are integers in $[0, W]$, and a source vertex $s \in V$. Their goal is to compute the shortest path distance $d_G(s,v)$ from $s$ to every vertex $v \in V$.

\begin{theorem}
\label{thm:SSSP}
There is a deterministic protocol $\mathcal{P}^{\mathsf{SSSP}}$ that computes $\SSSP$ using $O(n \log(nW))$ bits of communication.
\end{theorem}

We show that Dijkstra's algorithm can be simulated using $O(n\log(nW))$ bits of communication. For simplicity, we first recall a basic version of Dijkstra's algorithm without a priority queue.

\begin{algorithm}[H]
\caption{Dijkstra's Algorithm (without a priority queue)}
\label{alg:dijkstra}
\KwIn{A directed weighted graph $G = (V,E,w)$ and a source vertex $s$.}
\KwOut{The shortest-path distance $d(v)$ from $s$ to every $v \in V$.}
$S \gets \{s\}$\;
$d(s) \gets 0$\;
$d(v) \gets +\infty$ for each $v \in V \setminus \{s\}$\;
\For{$i = 1, \ldots, n-1$}{
    \For{$v \in (V \setminus S)$}{ \label{line:dijkstra-inner-for}
        Compute $\displaystyle \widetilde d(v) := \min_{u \in S, (u,v) \in E} \{d(u) + w(u,v)\}$ (or $\widetilde d(v) := +\infty$ if no such $u$ exists)\; \label{line:dijkstra-compute}
    }
    $v_i \gets \arg\min_{v \in V\setminus S} \widetilde d(v)$ (breaking ties lexicographically)\; \label{line:dijkstra-argmin}
    $d(v_i) \gets \widetilde d(v_i)$\; \label{line:dijkstra-set-d}
    $S \gets S \cup \{v_i\}$\;
}
\Return $d(v)$ for all $v \in V$\;
\end{algorithm}

\begin{proof}[Proof of \Cref{thm:SSSP}]
Let $E_A$ and $E_B$ denote the sets of edges held by Alice and Bob, respectively.
They maintain a consistent copy of the set $S$ and the array $d$, and simulate \Cref{alg:dijkstra} by replacing Lines~\ref{line:dijkstra-inner-for}--\ref{line:dijkstra-set-d} with the following process:
In iteration $i$,
\begin{enumerate}
    \item Alice locally computes $\widetilde d_A(v) := \min_{u \in S, (u, v)\in E_A}\{d(u) + w(u, v)\}$ for each vertex $v \in V \setminus S$, where $\widetilde d_A(v) := +\infty$ if no such $u$ exists.
          Alice sends the vertex $v^A_i$ minimizing $\widetilde d_A(v)$ (breaking ties lexicographically) and the corresponding $\widetilde d_A(v^A_i)$ to Bob.
          Bob does the same and sends $v^B_i$ and $\widetilde d_B(v^B_i)$ to Alice. 
    \item Alice and Bob then agree on $v_i \in \{v^A_i, v^B_i\}$ based on which of $\widetilde d_A(v^A_i)$ and $\widetilde d_B(v^B_i)$ is smaller (again breaking ties lexicographically).
    \item Then they set $d(v_i) := \min(\widetilde d_A(v^A_i), \widetilde d_B(v^B_i))$.
          (Note that it is possible that $v^A_i = v^B_i$ but $\widetilde d_A(v^A_i) \ne \widetilde d_B(v^B_i)$.)
\end{enumerate}

This produces the same choice of $v_i$ and the same value $d(v_i)$ as Lines~\ref{line:dijkstra-inner-for}--\ref{line:dijkstra-set-d}, because the minimum candidate over all edges from $S$ to $V\setminus S$ is the minimum of Alice's best candidate and Bob's best candidate.
Since edge weights lie in $[0, W]$, all finite distances are at most $(n - 1)W$, so each distance value can be encoded using $O(\log(nW))$ bits, where $+\infty$ can be represented using a special symbol.
Each iteration therefore uses $O(\log n + \log(nW)) = O(\log(nW))$ bits of communication, and there are $(n-1)$ iterations.
Hence, the total communication is $O(n\log(nW))$ bits.
\end{proof}

\section{Omitted Proofs from Section~\ref{sec:bipartite_matching}}
\label{sec:bipartite_matching_proofs}

\volumegap*
\begin{proof}
Recall that in a bipartite graph, the sizes of the (integral) maximum matching and minimum vertex cover are the same.
Moreover, the integrality gaps of the maximum matching and minimum vertex cover LPs are both $1$.

If $G$ has a perfect matching, then the minimum vertex cover of $G$ has size $\frac{n}{2}$.
Because any feasible point in $P^{\textnormal{VC}}(G)$ is a fractional vertex cover of size at most $\frac{n}{2} - \frac{1}{2}$, $P^{\textnormal{VC}}(G)$ must be empty.

If $G$ has no perfect matching, then $G$ has an integral vertex cover $y \in \{0,1\}^n$ of size $\frac{n}{2}-1$. We have $\sum y_v \le \frac{n}{2} - 1$ and $y_u + y_v \ge 1$ for all $(u,v) \in E$.
We show that $B \subseteq P^{\textnormal{VC}}(G)$ where
\begin{align*}
    B = \prod_{v \in V}
\begin{cases}
\left[\frac{1}{4n}, \frac{1}{2n}\right] & \text{if } y_v = 0, \\
\left[1 - \frac{1}{4n}, 1\right] & \text{if } y_v = 1.
\end{cases}
\end{align*}
Fix any $q \in B \subseteq [0, 1]^n$. For any edge $(u,v) \in E$, at least one of $y_u$ or $y_v$ is $1$, so $q_u + q_v \ge \frac{1}{4n} + \left(1 - \frac{1}{4n}\right) = 1$. Moreover, $\sum_v q_v \le \sum_v \left(y_v + \frac{1}{2n}\right) = \left(\frac{n}{2} - 1\right) + \frac{1}{2} = \frac{n}{2} - \frac{1}{2}$.
This shows that $q \in P^{\textnormal{VC}}(G)$.
Therefore, $B \subseteq P^{\textnormal{VC}}(G)$ and $\vol(P^{\textnormal{VC}}(G)) \ge \vol(B) = (4n)^{-n}$ as claimed.
\end{proof}

\subsection{Proof of Lemma~\ref{lem:discrete_BMM_progress}}
\label{sec:discrete_BMM_progress_proofs}

We first introduce the notion of a \emph{balanced polytope}.

\begin{definition}[$\beta$-balanced polytope]
Let $P \subseteq \mathbb{R}^n$ be a polytope. We say that $P$ is $\beta$-balanced if there exists a point $p \in P$ such that $d(p, F) \ge \beta$ for all facets $F$ of $P$, where $d(p,F)$ is the distance from $p$ to the supporting hyperplane of $F$.
\end{definition}

Given a bipartite graph $G = (V, E)$, recall that we define $P^{\textnormal{VC}}(G)$ and $Q^{\textnormal{VC}}(G)$ as follows:
\begin{align*}
P^{\textnormal{VC}}(G) &:=
\left\{ x \in \mathbb{R}^n \; \middle| \;
\begin{aligned}
\, & x_u + x_v \ge 1 && \forall (u, v) \in E, \\
\, & \textstyle\sum_{v\in V} x_v \le \frac{n}{2} - \frac{1}{2}, \\
\, & 0\le x_v \le 1 && \forall v \in V
\end{aligned}
\right\} \\
Q^{\textnormal{VC}}(G) &:=
\left\{ x \in \mathbb{R}^n \; \middle| \;
\begin{aligned}
\, & x_u + x_v \ge 1 && \forall (u, v) \in E, \\
\, & \textstyle\sum_{v\in V} x_v \le \frac{n}{2} - \frac{1}{2}, \\
\, & x_v \in \left\{0, \tfrac{1}{d}, \ldots, 1\right\} && \forall v \in V
\end{aligned}
\right\}.
\end{align*}

A key observation is that every non-empty vertex-cover polytope is $\frac{1}{8n}$-balanced.
\Cref{cor:P_vertex} follows directly from the proof of \Cref{lem:P_vertex}.

\begin{corollary}
\label{cor:P_vertex}
Let $G$ be a bipartite graph with no perfect matching and $\frac{n}{2}$ vertices on each side.
Then, $P^{\textnormal{VC}}(G)$ is $\frac{1}{8n}$-balanced.
\end{corollary}

The following lemma states that if the grid has sufficiently fine resolution, the volume of a balanced polytope can be well approximated by the (scaled) number of grid points it contains.
We prove \Cref{lem:approx_volume} in \Cref{subsec:proof_vol_approx}.

\begin{restatable}[Volume approximation for $\beta$-balanced polytopes]{lem}{approxvolume}
\label{lem:approx_volume}
Let $P \subseteq \mathbb{R}^n$ be a $\beta$-balanced polytope.
Let $d = \lceil \beta^{-1} n^{6} \rceil$.
Let $Q = P \cap \left(\mathbb{Z}/d\right)^{n}$ be the set of discretized points in $P$. Then
\begin{align*}
    1-n^{-3} \le \frac{d^{-n} \cdot |Q|}{\vol(P)} \le 1+n^{-3}.
\end{align*}
\end{restatable}

The following fact is immediate from the definitions of $P^{\textnormal{VC}}(G)$ and $Q^{\textnormal{VC}}(G)$.
\begin{fact}
$Q^{\textnormal{VC}}(G) = P^{\textnormal{VC}}(G) \cap (\mathbb{Z}/d)^n$.
\label{fact:Q_in_P}
\end{fact}

We are now ready to prove \Cref{lem:discrete_BMM_progress}.

\discreteBMMprogress*
\begin{proof}
Since $Q^{\textnormal{VC}}(H) \ne \emptyset$, by \Cref{fact:Q_in_P}, $P^{\textnormal{VC}}(H)$ is also non-empty.
By \Cref{lem:P_vertex}, this means $H$ has no perfect matching, and then by \Cref{cor:P_vertex}, $P^{\textnormal{VC}}(H)$ is $\frac{1}{8n}$-balanced.
By \Cref{lem:approx_volume},
\[
\vol\!\left(P^{\textnormal{VC}}(H)\right) \le \frac{d^{-n}}{1 - n^{-3}} \cdot \left|Q^{\textnormal{VC}}(H)\right|.
\]

Let $z$ be the center of mass of $P^{\textnormal{VC}}(H)$.
Then $z$ is a fractional vertex cover of $H$ of size $\sum_{v \in V} z_v \le \frac{n}{2} - \frac{1}{2}$.
Let $M$ be a perfect matching of $G$.
There exists an edge $f = (a,b) \in M$ that is not covered by $z$, i.e., $z_a + z_b < 1$.
Since $z$ is a fractional vertex cover of $H$, $f \in (G \setminus H)$.

Let $H' := H \cup \{f\}$.
Because $z$ violates the constraint of edge $f$, Gr\"unbaum's Theorem~\cite{Grunbaum1960} gives
\[
\vol\!\left(P^{\textnormal{VC}}(H')\right) \le \left(1 - \frac{1}{e}\right) \cdot \vol\!\left(P^{\textnormal{VC}}(H)\right).
\]

We consider two cases.
If $P^{\textnormal{VC}}(H') = \emptyset$, then $Q^{\textnormal{VC}}(H') = \emptyset$ by \Cref{fact:Q_in_P}, and the lemma holds.

If $P^{\textnormal{VC}}(H') \ne \emptyset$, then by \Cref{lem:P_vertex}, $H'$ has no perfect matching. By \Cref{cor:P_vertex}, $P^{\textnormal{VC}}(H')$ is $\frac{1}{8n}$-balanced.
By \Cref{lem:approx_volume}, we have
\[
\vol\!\left(P^{\textnormal{VC}}(H')\right) \ge \frac{d^{-n}}{1 + n^{-3}} \cdot \left|Q^{\textnormal{VC}}(H')\right|,
\]
and consequently,
\[
\frac{\left|Q^{\textnormal{VC}}(H')\right|}{\left|Q^{\textnormal{VC}}(H)\right|} \le \frac{1 + n^{-3}}{1 - n^{-3}} \cdot \frac{\vol\!\left(P^{\textnormal{VC}}(H')\right)}{\vol\!\left(P^{\textnormal{VC}}(H)\right)} \le \frac{1 + n^{-3}}{1 - n^{-3}} \cdot \left(1 - \frac{1}{e}\right) \le 0.9,
\]
where the last inequality holds for all $n \ge 2$.
Thus $f$ satisfies $\left|Q^{\textnormal{VC}}(H \cup \{f\})\right| \le 0.9 \cdot \left|Q^{\textnormal{VC}}(H)\right|$, as claimed.
\end{proof}

\subsection{Proof of Lemma~\ref{lem:approx_volume}}
\label{subsec:proof_vol_approx}

For any pair of points $a, b \in \mathbb{R}^n$, we denote by $ab$ the line segment with endpoints $a$ and $b$, and by $|ab|$ its Euclidean length.
For an affine hyperplane $H$, let $d(a, H)$ denote the distance from point $a$ to $H$.
This is equal to $|az|$, where $z$ is the orthogonal projection of $a$ onto $H$.
When $F$ is a facet of a polytope, we write $d(a,F)$ for the distance from $a$ to the
supporting hyperplane of $F$.

\begin{fact}
Let $H$ be an affine hyperplane, and let $a,b\in\mathbb{R}^n$ be such that
the segment $ab$ intersects $H$. Then $d(a,H)\le |ab|$ and $d(b,H)\le |ab|$.
\label{fact:oppo_facet_dist}
\end{fact}

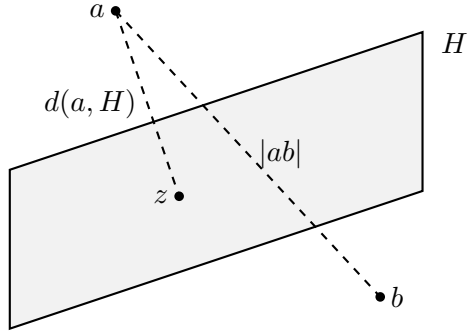
\begin{figure}[ht]
\begin{center}
\begin{tikzpicture}[scale=1.4, thick]
  \coordinate (H1) at (-2,1.5);
  \coordinate (H2) at (2,2);
  \draw[fill=gray!10] (H1) -- ++(0,-1.5) -- ++(3.9,1.3) -- ++(0,1.5) -- cycle;
  \node at (2.2,2.7) {$H$};

  \coordinate (a) at (-1,3);
  \node[left] at (a) {$a$};
  \filldraw[black] (a) circle (1pt);

  \coordinate (z) at (-0.4,1.25);
  \node[left] at (z) {$z$};
  \filldraw[black] (z) circle (1pt);

  \draw[dashed] (a) -- (z);
  \node[left] at ($(a)!0.5!(z)$) {$d(a,H)$};

  \coordinate (b) at (1.5,0.3);
  \node[right] at (b) {$b$};
  \filldraw[black] (b) circle (1pt);

  \draw[dashed] (a) -- (b);
  \node[right] at ($(a)!0.5!(b)$) {$|ab|$};

\end{tikzpicture}
\caption{An illustration of \Cref{fact:oppo_facet_dist}.}
\label{fig:distances}
\end{center}
\end{figure}

Now we prove our volume approximation lemma.

\approxvolume*
\begin{proof}
We first show the upper bound: $\frac{d^{-n} \cdot |Q|}{\vol(P)} \le 1 + n^{-3}$.
For each discrete point $q \in Q$ define an $n$-dimensional cube $B_q$ of side length $\frac{1}{d}$, centered at $q$:
\begin{align*}
B_q = \prod_{v \in V} \left[q_v - \frac{1}{2d}, \; q_v + \frac{1}{2d} \right].
\end{align*}
Let $\mathcal{B} = \bigcup_{q \in Q} B_q$ be the union of all such cubes. Since $Q$ is a subset of the $\frac{1}{d}$-grid, the cubes $\{B_q\}_{q\in Q}$ are interior-disjoint. Thus, $\vol(\mathcal{B}) = d^{-n} \cdot |Q|$. We will also use the following fact: for any $x \in B_q$, we have $|xq| \le \frac{n}{d}$.

\begin{figure}[ht]
\begin{center}
\begin{tikzpicture}[scale=0.9, thick]

  \coordinate (11) at (1,1);
  \coordinate (12) at (1,2);
  \coordinate (13) at (1,3);
  \coordinate (14) at (1,4);
  \coordinate (15) at (1,5);
  \coordinate (16) at (1,6);
  
  \coordinate (21) at (2,1);
  \coordinate (22) at (2,2);
  \coordinate (23) at (2,3);
  \coordinate (24) at (2,4);
  \coordinate (25) at (2,5);
  \coordinate (26) at (2,6);
  
  \coordinate (31) at (3,1);
  \coordinate (32) at (3,2);
  \coordinate (33) at (3,3);
  \coordinate (34) at (3,4);
  \coordinate (35) at (3,5);
  \coordinate (36) at (3,6);
  
  \coordinate (41) at (4,1);
  \coordinate (42) at (4,2);
  \coordinate (43) at (4,3);
  \coordinate (44) at (4,4);
  \coordinate (45) at (4,5);
  \coordinate (46) at (4,6);
  
  \coordinate (51) at (5,1);
  \coordinate (52) at (5,2);
  \coordinate (53) at (5,3);
  \coordinate (54) at (5,4);
  \coordinate (55) at (5,5);
  \coordinate (56) at (5,6);

  \coordinate (61) at (6,1);
  \coordinate (62) at (6,2);
  \coordinate (63) at (6,3);
  \coordinate (64) at (6,4);
  \coordinate (65) at (6,5);
  \coordinate (66) at (6,6);

  \coordinate (v1) at (1.2, 3.4);
  \coordinate (v2) at (2.8, 1.1);
  \coordinate (v3) at (4.8, 2.15);
  \coordinate (v4) at (5.85, 4.35);
  \coordinate (v5) at (3.5, 5.8);

  \draw[gray!80] (13) -- (14);
  \draw[gray!80] (22) -- (23);
  \draw[gray!80] (23) -- (24);
  \draw[gray!80] (24) -- (25);
  \draw[gray!80] (31) -- (32);
  \draw[gray!80] (32) -- (33);
  \draw[gray!80] (33) -- (34);
  \draw[gray!80] (34) -- (35);
    \draw[gray!80] (35) -- (36);
  \draw[gray!80] (41) -- (42);
  \draw[gray!80] (42) -- (43);
  \draw[gray!80] (43) -- (44);
  \draw[gray!80] (44) -- (45);
    \draw[gray!80] (45) -- (46);
  \draw[gray!80] (52) -- (53);
  \draw[gray!80] (53) -- (54);
  \draw[gray!80] (54) -- (55);
  \draw[gray!80] (64) -- (65);

  \draw[gray!80] (31) -- (41);
  \draw[gray!80] (22) -- (32);
  \draw[gray!80] (32) -- (42);
  \draw[gray!80] (42) -- (52);
    \draw[gray!80] (13) -- (23);
  \draw[gray!80] (23) -- (33);
  \draw[gray!80] (33) -- (43);
  \draw[gray!80] (43) -- (53);
    \draw[gray!80] (14) -- (24);
  \draw[gray!80] (24) -- (34);
  \draw[gray!80] (34) -- (44);
  \draw[gray!80] (44) -- (54);
    \draw[gray!80] (54) -- (64);
  \draw[gray!80] (25) -- (35);
  \draw[gray!80] (35) -- (45);
  \draw[gray!80] (45) -- (55);
    \draw[gray!80] (55) -- (65);

  \draw[gray!80] (36) -- (46);

    \draw (v1) -- (v2);
  \draw (v2) -- (v3);
  \draw (v3) -- (v4);
  \draw (v4) -- (v5);
  \draw (v5) -- (v1);

  \node at (4.5,5.5) {$P$};

  \filldraw[black] (1.5,3.5) circle (0.5pt);
  \filldraw[black] (2.5,2.5) circle (0.5pt);
  \filldraw[black] (2.5,3.5) circle (0.5pt);
  \filldraw[black] (2.5,4.5) circle (0.5pt);
  \filldraw[black] (3.5,1.5) circle (0.5pt);
  \filldraw[black] (3.5,2.5) circle (0.5pt);
  \filldraw[black] (3.5,3.5) circle (0.5pt); 
  \filldraw[black] (3.5,4.5) circle (0.5pt);
  \filldraw[black] (3.5,5.5) circle (0.5pt);
  \filldraw[black] (4.5,2.5) circle (0.5pt);
  \filldraw[black] (4.5,3.5) circle (0.5pt);
  \filldraw[black] (4.5,4.5) circle (0.5pt);
  \filldraw[black] (5.5,4.5) circle (0.5pt);

  \coordinate (x11) at (8.5,1);
  \coordinate (x12) at (8.5,2);
  \coordinate (x13) at (8.5,3);
  \coordinate (x14) at (8.5,4);
  \coordinate (x15) at (8.5,5);
  \coordinate (x16) at (8.5,6);
  
  \coordinate (x21) at (9.5,1);
  \coordinate (x22) at (9.5,2);
  \coordinate (x23) at (9.5,3);
  \coordinate (x24) at (9.5,4);
  \coordinate (x25) at (9.5,5);
  \coordinate (x26) at (9.5,6);
  
  \coordinate (x31) at (10.5,1);
  \coordinate (x32) at (10.5,2);
  \coordinate (x33) at (10.5,3);
  \coordinate (x34) at (10.5,4);
  \coordinate (x35) at (10.5,5);
  \coordinate (x36) at (10.5,6);
  
  \coordinate (x41) at (11.5,1);
  \coordinate (x42) at (11.5,2);
  \coordinate (x43) at (11.5,3);
  \coordinate (x44) at (11.5,4);
  \coordinate (x45) at (11.5,5);
  \coordinate (x46) at (11.5,6);
  
  \coordinate (x51) at (12.5,1);
  \coordinate (x52) at (12.5,2);
  \coordinate (x53) at (12.5,3);
  \coordinate (x54) at (12.5,4);
  \coordinate (x55) at (12.5,5);
  \coordinate (x56) at (12.5,6);

  \coordinate (x61) at (13.5,1);
  \coordinate (x62) at (13.5,2);
  \coordinate (x63) at (13.5,3);
  \coordinate (x64) at (13.5,4);
  \coordinate (x65) at (13.5,5);
  \coordinate (x66) at (13.5,6);

  \coordinate (xv1) at (8.7, 3.4);
  \coordinate (xv2) at (10.3, 1.1);
  \coordinate (xv3) at (12.3, 2.15);
  \coordinate (xv4) at (13.35, 4.35);
  \coordinate (xv5) at (11, 5.8);

  \draw[gray!80] (x13) -- (x14);
  \draw[gray!80] (x22) -- (x23);
  \draw[gray!80] (x23) -- (x24);
  \draw[gray!80] (x24) -- (x25);
  \draw[gray!80] (x31) -- (x32);
  \draw[gray!80] (x32) -- (x33);
  \draw[gray!80] (x33) -- (x34);
  \draw[gray!80] (x34) -- (x35);
    \draw[gray!80] (x35) -- (x36);
  \draw[gray!80] (x41) -- (x42);
  \draw[gray!80] (x42) -- (x43);
  \draw[gray!80] (x43) -- (x44);
  \draw[gray!80] (x44) -- (x45);
    \draw[gray!80] (x45) -- (x46);
  \draw[gray!80] (x52) -- (x53);
  \draw[gray!80] (x53) -- (x54);
  \draw[gray!80] (x54) -- (x55);
  \draw[gray!80] (x64) -- (x65);

  \draw[gray!80] (x31) -- (x41);
  \draw[gray!80] (x22) -- (x32);
  \draw[gray!80] (x32) -- (x42);
  \draw[gray!80] (x42) -- (x52);
    \draw[gray!80] (x13) -- (x23);
  \draw[gray!80] (x23) -- (x33);
  \draw[gray!80] (x33) -- (x43);
  \draw[gray!80] (x43) -- (x53);
    \draw[gray!80] (x14) -- (x24);
  \draw[gray!80] (x24) -- (x34);
  \draw[gray!80] (x34) -- (x44);
  \draw[gray!80] (x44) -- (x54);
    \draw[gray!80] (x54) -- (x64);
  \draw[gray!80] (x25) -- (x35);
  \draw[gray!80] (x35) -- (x45);
  \draw[gray!80] (x45) -- (x55);
    \draw[gray!80] (x55) -- (x65);

  \draw[gray!80] (x36) -- (x46);

    \draw (xv1) -- (xv2);
  \draw (xv2) -- (xv3);
  \draw (xv3) -- (xv4);
  \draw (xv4) -- (xv5);
  \draw (xv5) -- (xv1);

  \node at (12,5.5) {$P$};

  \coordinate (xx) at (10.9,3.35);

  \coordinate (xu1) at (7.82, 3.42);
  \coordinate (xu2) at (10.06, 0.20);
  \coordinate (xu3) at (12.86, 1.67);
  \coordinate (xu4) at (14.33, 4.75);
  \coordinate (xu5) at (11.04, 6.78);

  \draw[dashed][gray!80] (xx) -- (xu1);
  \draw[dashed][gray!80] (xx) -- (xu2);
  \draw[dashed][gray!80] (xx) -- (xu3);
  \draw[dashed][gray!80] (xx) -- (xu4);
  \draw[dashed][gray!80] (xx) -- (xu5);

    \node at (12.44,6.36) {$P^+$};

    \draw (xu1) -- (xu2);
  \draw (xu2) -- (xu3);
  \draw (xu3) -- (xu4);
  \draw (xu4) -- (xu5);
  \draw (xu5) -- (xu1);
  
    \node[above] at (10.75, 3.35) {$p$};
  \filldraw[black] (10.9,3.35) circle (1pt);

\end{tikzpicture}
\caption{Left: An illustration of a polytope $P$ in the ambient space $\mathbb{R}^2$. Each ``$\cdot$'' represents a discretized point $q \in Q$, and each overlaid square grid denotes the corresponding 2-dimensional cube $B_q$. The union of all such cubes (i.e., grids) forms the set $\mathcal{B} = \bigcup_q B_q$. Right: An illustration of a polytope $P^+$ obtained by scaling $P$ from a point $p$ using an appropriate scaling factor.}
\label{fig:grid}
\end{center}
\end{figure}
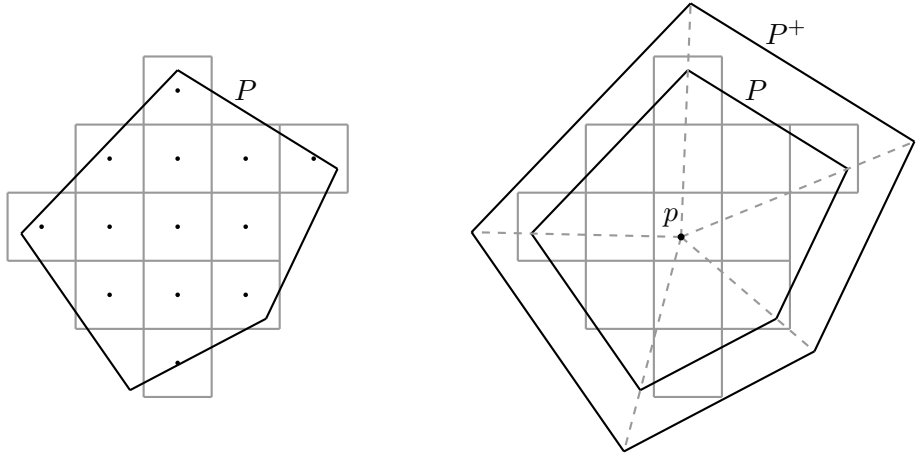

Let $p \in P$ be a point such that $d(p, F) \ge \beta$ for all facets $F$ of $P$. Such a point exists by the assumption that $P$ is $\beta$-balanced. Define $\varepsilon = \frac{n}{\beta d}$ and
\begin{align*}
P^+ := \left\{ p + (1+\varepsilon) \cdot (x - p) \; \middle| \; x \in P \right\}.
\end{align*}
That is, $P^+$ is the result of scaling $P$ by a factor of $1+\varepsilon$ from the point $p$. From this construction, the following fact about volumes is immediate:
\begin{align*}
\vol(P^+) = (1+\varepsilon)^n \cdot \vol(P).
\end{align*}

We claim that $\mathcal{B} \subseteq P^+$. To prove this, it suffices to show that for every $q\in Q$ and every $x\in B_q$, we have $x\in P^+$. If $x \in P$, we are done as $P \subseteq P^+$. Otherwise, let us consider $x \notin P$.

\begin{figure}[ht]
\begin{center}
\begin{tikzpicture}[scale=1, thick]


\foreach \i in {1,...,6} {
    \draw[gray!25, very thin] (\i,1) -- (\i,6);
    \draw[gray!25, very thin] (1,\i) -- (6,\i);
}

\fill[gray!15] (2,4) rectangle (3,5);
\draw[gray!70] (2,4) rectangle (3,5);

\coordinate (v1) at (1.2,3.4);
\coordinate (v2) at (2.8,1.1);
\coordinate (v3) at (4.8,2.15);
\coordinate (v4) at (5.85,4.35);
\coordinate (v5) at (3.5,5.8);

\draw (v1) -- (v2) -- (v3) -- (v4) -- (v5) -- cycle;
\node at (2.9,5.5) {$F$};

\coordinate (p) at (3.4,3.35);
\coordinate (x) at (2.4,4.9);
\coordinate (y) at (2.50,4.75);
\coordinate (q) at (2.5,4.5);

\draw (p) -- (x);

\filldraw[black] (p) circle (1pt);
\node[right] at (p) {$p$};

\filldraw[black] (x) circle (0.7pt);
\node[left] at (x) {$x$};

\filldraw[black] (y) circle (0.7pt);
\node[right] at (y) {$y$};

\filldraw[black] (q) circle (0.7pt);
\node[below left] at (q) {$q$};


\coordinate (l) at (8,3.7);
\coordinate (r) at (13.5,3.7);
\draw (l) -- (r);
\node[right] at (r) {$F$};

\coordinate (pR) at (9,1.2);
\coordinate (xR) at (12.5,5.4);
\coordinate (yR) at (11.083,3.7);

\coordinate (pF) at (9,3.7);
\coordinate (xF) at (12.5,3.7);

\draw (pR) -- (xR);

\draw[dashed] (pR) -- (pF);
\draw[dashed] (xR) -- (xF);

\filldraw[black] (pR) circle (0.7pt);
\node[right] at (pR) {$p$};

\filldraw[black] (xR) circle (0.7pt);
\node[right] at (xR) {$x$};

\filldraw[black] (yR) circle (0.7pt);
\node[below] at (yR) {$y$};

\node[left] at ($(pR)!0.5!(pF)$) {$d(p,F)$};
\node[right] at ($(xR)!0.5!(xF)$) {$d(x,F)$};

\node[right] at ($(pR)!0.5!(yR)$) {$|py|$};
\node[left] at ($(xR)!0.5!(yR)$) {$|xy|$};

\end{tikzpicture}
\caption{Shown on the left is an illustration of the points $p,x,y$, and $q$.
Shown on the right are the similar triangles that imply $\frac{|py|}{|xy|} = \frac{d(p,F)}{d(x,F)}$.}
\label{fig:pqxy}
\end{center}
\end{figure}
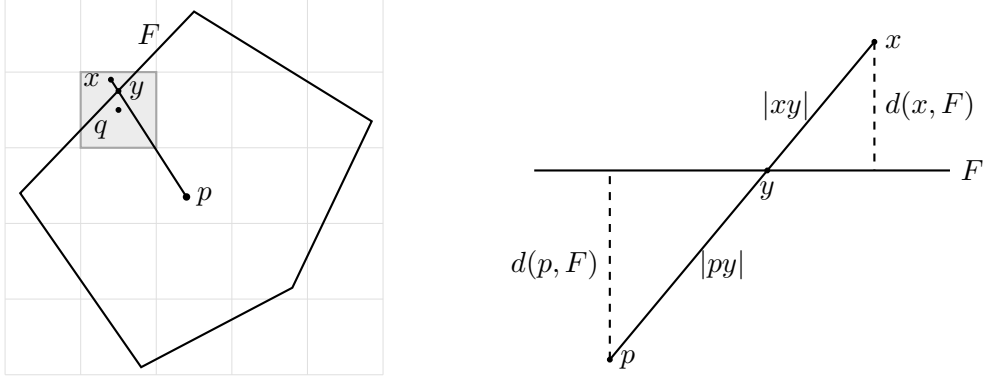

Let $y$ be the first point where the segment $px$ intersects $\partial P$.
Choose a facet $F$ whose supporting inequality is tight at $y$ and violated by
$x$. To show that $x \in P^+$, it suffices to show $|px| \le (1+\varepsilon) \cdot |py|$ which is equivalent to $|py| \ge \frac{|xy|}{\varepsilon}.$ This is true by the following derivation:
\begin{align*}
\underbrace{\frac{|py|}{|xy|} = \frac{d(p, F)}{d(x, F)}}_{\substack{\text{by similar triangles,} \\ \text{see \Cref{fig:pqxy}}}} \ge \frac{\beta}{|xq|} \ge \frac{\beta d}{n} = \frac{1}{\varepsilon},
\end{align*}
where the inequalities use the $\beta$-balancedness of $P$, the fact that
$q\in Q\subseteq P$ while $x$ violates the facet-defining halfspace of $F$
(and hence \Cref{fact:oppo_facet_dist} gives $d(x,F)\le |xq|$), and $|xq|\le n/d$.

It then follows that
\begin{align*}
d^{-n} \cdot |Q| = \vol(\mathcal{B}) \le \vol(P^+) = (1+\varepsilon)^n \cdot \vol(P).
\end{align*}
Rearranging yields,
\[
\frac{d^{-n} \cdot |Q|}{\vol(P)} \le (1+\varepsilon)^n = \left(1 + \frac{n}{\beta d}\right)^n \le \left(1 + n^{-5}\right)^n \le \exp(n^{-4}) \le 1 + 2n^{-4} \le 1 + n^{-3},
\]
where we use $d = \lceil \beta^{-1} n^{6} \rceil \ge \beta^{-1} n^6$, $1 + x \le e^x \le 1 + 2x$ for $x \in [0, 1]$,  and $n \ge 2$.

For the lower bound, define
\[
P^- := \left\{ p+(1-\varepsilon)(x-p) \;\middle|\; x\in P \right\}.
\]
We claim $P^- \subseteq \mathcal{B}$. Fix $x\in P^-$, and let $q$ be a
nearest point to $x$ on the $\frac{1}{d}$-grid, so that $x\in B_q$ and
$|xq|\le \frac{n}{d}$. It remains to show that $q\in P$, which would imply
$q\in Q$ and hence $x\in\mathcal{B}$.

For every facet $F$ of $P$, the point $x$ is at distance at least
\[
    \varepsilon d(p,F) \ge \varepsilon\beta = \frac{n}{d}
\]
from the supporting hyperplane of $F$. If $q\notin P$, then $q$ violates some
facet-defining inequality. Let $F$ be the corresponding facet, let $H_F$ be its
supporting hyperplane, and let $H_F^+$ be the facet-defining halfspace containing
$P$. Since $x\in P$ and $q\notin H_F^+$, the segment $xq$ intersects $H_F$.
By \Cref{fact:oppo_facet_dist},
\[
    d(x,F) = d(x,H_F) \le |xq| \le \frac{\sqrt n}{2d} < \frac{n}{d},
\]
where the second inequality follows because $q$ is a nearest $\frac{1}{d}$-grid point to $x$.
This contradicts the lower bound $d(x,F)\ge \frac{n}{d}$. Therefore, $q\in P$, and
hence $P^-\subseteq\mathcal{B}$. Consequently,
\[
d^{-n} \cdot |Q| = \vol(\mathcal{B}) \ge \vol(P^-)
= (1-\varepsilon)^n \cdot \vol(P).
\]

It follows that
\[
\frac{d^{-n} \cdot |Q|}{\vol(P)} \ge (1-\varepsilon)^n = \left(1 - \frac{n}{\beta d}\right)^n \ge \left(1 - n^{-5}\right)^n \ge 1 - n(n^{-5}) = 1 - n^{-4} \ge 1 - n^{-3}. \qedhere \]
\end{proof}

\end{document}